\documentclass[11pt,a4paper,USenglish]{article}

\usepackage[left=1in,right=1in,top=1in,bottom=1in]{geometry}

\usepackage{amsthm}
\usepackage{amsmath}
\usepackage{subcaption}
\usepackage{forest}
\usepackage{amssymb}
\usepackage{mathtools}

\usepackage[bookmarks]{hyperref}

\usepackage[utf8]{inputenc}
\usepackage[OT4]{fontenc}
\usepackage{xcolor} 
\usepackage{cite} 
\usepackage{enumerate} 
\usepackage{microtype} 
\usepackage{stmaryrd}
\usepackage{xspace}
\usepackage{enumitem}
\usepackage{comment}

\newtheorem{theorem}{Theorem}[section]
\newtheorem{lemma}[theorem]{Lemma}

\newcommand{\RR}{\mathbb{R}}
\newcommand{\NN}{\mathbb{N}}

\newcommand{\ZZ}{\mathbb{Z}}

\newcommand{\Ex}{\mathbb{E}}

\newcommand{\volsel}{\textsc{Volume Selection}}
\newcommand{\HSS}{\textsc{VolSel}}
\newcommand{\bbox}{\textsc{box}}
\newcommand{\vol}{\textsc{vol}}
\newcommand{\diff}{\textsc{diff}}
\newcommand{\unio}{\mathcal{U}}

\newcommand{\polylog}{\textup{polylog}}

\newcommand{\optcomp}{\Phi_{\rm comp}}
\newcommand{\optfree}{\Phi_{\rm free}}
\newcommand{\dpcomp}{\Psi_{\rm comp}}

\newcommand{\D}{\mathcal{D}}

\newcommand{\eps}{\varepsilon}

\newcommand{\bx}{\bar{x}}
\newcommand{\by}{\bar{y}}
\newcommand{\bl}{{\bar{\ell}}}

\newcommand{\ignore}[1]{}

\def\DEF#1{\textbf{\emph{#1}}}

\title{Maximum Volume Subset Selection for Anchored Boxes}

\author{
Karl Bringmann\thanks{Max Planck Institute for Informatics, Saarland Informatics Campus, Saarbr\"ucken, Germany.} \and
Sergio Cabello\thanks{Department of Mathematics, IMFM, and Department of Mathematics, FMF, University of Ljubljana, Slovenia. Supported by the Slovenian Research Agency, program P1-0297 and project L7-5459.} \and
Michael T.M. Emmerich\thanks{Leiden Institute of Advanced Computer Science (LIACS), Leiden University, the Netherlands.}
}

\date{}

\begin{document}
\setcounter{page}{0}
\maketitle
\thispagestyle{empty}

\begin{abstract}
	Let $B$ be a set of $n$ axis-parallel boxes in $\RR^d$ 
    such that each box has a corner at the origin and the other corner
    in the positive quadrant of $\RR^d$, and let $k$ be a positive integer.  
    We study the problem of selecting $k$ boxes in $B$ that maximize
    the volume of the union of the selected boxes.
    This research is motivated by applications in skyline queries for databases and in multicriteria optimization, where the problem is known as the \emph{hypervolume subset 
    selection problem}.
    It is known that the problem can be solved in polynomial time in the plane, 
    while the best known running time in any dimension $d \ge 3$ is $\Omega\big(\binom{n}{k}\big)$.
    We show that:
    \begin{itemize}
    \item The problem is NP-hard already in 3 dimensions. 
    \item In 3 dimensions, we break the bound $\Omega\big(\binom{n}{k}\big)$, by providing an $n^{O(\sqrt{k})}$ algorithm.
    \item For any constant dimension $d$, we present an efficient polynomial-time approximation scheme. 
    \end{itemize}
\end{abstract}

\newpage

\section{Introduction}
An \emph{anchored box} is an orthogonal range of the form $\bbox(p) := [0,p_1]\times \ldots \times [0,p_d] \subset \RR_{\ge 0}^d$, spanned by the point $p \in \RR_{>0}^d$. 
This paper is concerned with the problem \emph{\volsel}: Given a set $P$ of $n$ points in $\RR_{>0}^d$, select $k$ points in $P$ maximizing the volume of the union of their anchored boxes. That is, we want to compute
$$ \HSS(P,k) := \max_{S\subseteq P, \, |S| = k} \vol \Big(\bigcup_{p \in S} \bbox(p) \Big),$$
as well as a set $S^* \subseteq P$ of size $k$ realizing this value.
Here, $\vol$ denotes the usual volume.

\paragraph*{Motivation}
This geometric problem is of key importance in the context of multicriteria optimization and decision analysis, where it is known as the \emph{hypervolume subset selection problem (HSSP)} \cite{ABB+09, ABZ12, Bader09,KFP+14,bringmann2014generic,BFK14}. In this context, the points in $P$ correspond to solutions of an optimization problem with $d$ objectives, and the goal is to find a small subset of $P$ that ``represents'' the set $P$ well. The quality of a representative subset $S \subseteq P$ is measured by the volume of the union of the anchored boxes spanned by points in $S$; this is also known as the \emph{hypervolume indicator}~\cite{zitzler2003performance}. Note that with this quality indicator, finding the optimal size-$k$ representation is equivalent to our problem $\HSS(P,k)$.
In applications, such bounded-size representations are required in archivers for non-dominated sets \cite{KCF03} and for multicriteria optimization algorithms and heuristics \cite{ABZ12,bringmann2010efficient,BNE07}.\footnote{We remark that in these applications the anchor point is often not the origin, however, by a simple translation we can move our anchor point from $(0,\ldots,0)$ to any other point in $\RR^d$.}
Besides, the problem has recently received attention in the context of skyline operators in databases \cite{EDY15}.

In 2 dimensions, the problem can be solved in polynomial time~\cite{ABB+09,BFK14, KFP+14}, which is used in applications such as analyzing benchmark functions~\cite{ABB+09} and efficient postprocessing of multiobjective algorithms~\cite{bringmann2014generic}. A natural question is whether efficient algorithms also exist in dimension $d \ge 3$, and thus whether these applications can be pushed beyond two objectives.

\smallskip
In this paper, we answer this question negatively, by proving that \volsel\ is NP-hard already in 3 dimensions. 
We then consider the question whether the previous $\Omega(\binom{n}{k})$ bound can be improved, which we answer affirmatively in 3 dimensions.
Finally, for any constant dimension, we improve the best-known $(1-1/e)$-approximation to an efficient polynomial-time approximation scheme (EPTAS).
See Section~\ref{sec:ourresults} for details.

\subsection{Further Related Work}

\paragraph*{Klee's Measure Problem}
To compute the volume of the union of $n$ (not necessarily anchored) axis-aligned boxes in $\RR^d$ is known as Klee's measure problem. 
The fastest known algorithm takes time\footnote{In $O$-notation, we always assume $d$ to be a constant, and $\log(x)$ is to be understood as $\max\{1, \log(x)\}$.} $O(n^{d/2})$, which can be improved to $O(n^{d/3} \polylog(n))$ if all boxes are cubes~\cite{chan2013klee}. By a simple reduction~\cite{bringmann2013bringing}, the same running time as on cubes can be obtained on anchored boxes, which can be improved to $O(n \log n)$ for $d \le 3$~\cite{beume2009complexity}.
These results are relevant to this paper because Klee's measure problem on anchored boxes (spanned by the points in~$P$) is a special case of \volsel\ (by calling $\HSS(P,|P|)$). 

Chan~\cite{CHAN2010243} gave a reduction from $k$-Clique to Klee's measure problem in $2k$ dimensions. This proves NP-hardness of Klee's measure problem when $d$ is part of the input (and thus $d$ can be as large as $n$). Moreover, since $k$-Clique has no $f(k) \cdot n^{o(k)}$-time algorithm under the Exponential Time Hypothesis~\cite{chen2004linear}, Klee's measure problem has no $f(d) \cdot n^{o(d)}$-time algorithm under the same assumption. The same hardness results also hold for Klee's measure problem on anchored boxes, by a reduction in~\cite{bringmann2013bringing} (NP-hardness was first proven in~\cite{bringmann2012approximating}).

Finally, we mention that Klee's measure problem has a very efficient randomized $(1 \pm \eps)$-approximation algorithm in time $O(n \log(1/\delta) /\eps^2)$ with error probability $\delta$~\cite{BRINGMANN2010601}.

\paragraph*{Known Results for Volume Selection}
As mentioned above, 2-dimensional \volsel\ can be solved in polynomial time; the initial $O(k n^2)$ algorithm~\cite{ABB+09} was later improved to $O((n-k)k + n \log n)$~\cite{BFK14,KFP+14}.
In higher dimensions, by enumerating all size-$k$ subsets and solving an instance of Klee's measure problem on anchored boxes for each one, there is an $O\big( \binom{n}{k} k^{d/3} \polylog(k)\big)$ algorithm. For small $n-k$, this can be improved to $O(n^{d/2} \log n + n^{n-k})$~\cite{bringmann2010efficient}.
\volsel\ is NP-hard when $d$ is part of the input, since the same holds already for Klee's measure problem on anchored boxes. However, this does not explain the exponential dependence on~$k$ for constant~$d$.

Since the volume of the union of boxes is a submodular function (see, e.g.,~\cite{ulrich2012bounding}), the greedy algorithm for submodular function maximization~\cite{nemhauser1978analysis} yields a $(1-1/e)$-approximation of $\HSS(P,k)$. This algorithm solves $O(n k)$ instances of Klee's measure problem on at most $k$ anchored boxes, and thus runs in time $O(n k^{d/3+1} \polylog(k))$. Using~\cite{BRINGMANN2010601}, this running time improves to $O(n k^2 \log(1/\delta) / \eps^2)$, at the cost of decreasing the approximation ratio to $1-1/e-\eps$ and introducing an error probability $\delta$.
See~\cite{GFP16} for related results in $3$ dimensions.

A problem closely related to \volsel\ is \emph{\textsc{Convex Hull Subset Selection}}: Given $n$ points in $\RR^d$, select $k$ points that maximize the volume of their convex hull. For this problem, NP-hardness was recently announced in the case $d=3$~\cite{RBB+16}.

\subsection{Our Results}
\label{sec:ourresults}

In this paper we push forward the understanding of \volsel. 
We prove that \volsel\ is NP-hard already for $d=3$ (Section~\ref{sec:hard}). Previously, NP-hardness was only known when $d$ is part of the input and thus can be as large as $n$. 
Moreover, this establishes \volsel\ as another example for problems that can be solved in polynomial time in the plane but are NP-hard in three or more dimensions (see also \cite{barahona1982computational,MitchellS04}).

In the remainder, we focus on the regime where $d \ge 3$ is a constant and $k \ll n$. 
All known algorithms (explicitly or implicitly) enumerate all size-$k$ subsets of the input set $P$ and thus take time $\Omega\big( \binom{n}{k} \big) = n^{\Omega(k)}$. In 3 dimensions, we break this time bound by providing an $n^{O(\sqrt{k})}$ algorithm (Section~\ref{sec:exact}). To this end, we project the 3-dimensional \volsel\ to a 2-dimensional problem and then use planar separator techniques.

Finally, in Section~\ref{sec:ptas} we design an EPTAS for \volsel. More precisely, we present a $(1-\eps)$-approximation algorithm running in time $O(n \cdot \eps^{-d} (\log n + k + 2^{O(\eps^{-2} \log 1/\eps)^d}))$, for any constant dimension $d$. Note that the ``combinatorial explosion'' is restricted to $d$ and $\eps$; for any constant $d,\eps$ the algorithm runs in time $O(n (k+\log n))$. This improves the previously best-known $(1-1/e)$-approximation, even in terms of running time.

\section{Preliminaries}
\label{sec:convention}

All boxes considered in the paper are axis-parallel and anchored at the origin. For points $p = (p_1,\ldots,p_d), \, q = (q_1,\ldots,q_d) \in \RR^d$, we say that $p$ \emph{dominates} $q$ if $p_i \ge q_i$ for all $1 \le i \le d$. For $p = (p_1,\ldots,p_d) \in \RR_{>0}^d$, we let 
$ \bbox(p) := [0,p_1] \times \ldots \times [0,p_d]$.
Note that $\bbox(p)$ is the set of all points $q \in \RR_{\ge 0}^d$ that are dominated by $p$.
A \emph{point set} $P$ is a set of points in $\RR_{> 0}^d$. We denote the union $\bigcup_{p \in P} \bbox(p)$ by $\unio(P)$. The usual Euclidean volume is denoted by $\vol$. With this notation, we set
\[ \mu(P) := \vol(\unio(P)) = \vol \Big( \bigcup_{p \in P} \bbox(p) \Big) = \vol \Big( \bigcup_{p \in P} [0,p_1] \times \ldots \times [0,p_d] \Big) . \]
We study \volsel: Given a point set $P$ of size $n$ and $0 \le k \le n$, compute 
\[ \HSS(P,k) := \max_{S \subseteq P, \, |S| = k} \mu(S). \]
Note that we can relax the requirement $|S| = k$ to $|S| \le k$ without changing this value.

\section{Hardness in 3 Dimensions}
\label{sec:hard}

We consider the following decision variant of 3-dimensional \volsel.

\begin{quote}
	{\sc 3d Volume Selection}\\
	\emph{Input:} A triple $(P,k,V)$, where $P$ is a set of points in $\RR_{>0}^3$,
		$k$ is a positive integer and $V$ is a positive real value.\\
	\emph{Question:} Is there a subset $Q\subseteq P$ of $k$ points such that
		$\mu(Q)\ge V$?
\end{quote}

We are going to show that the problem is NP-complete.
First, we show that an intermediate problem about selecting a large
independent set in a given induced subgraph of the triangular grid is NP-hard.
The reduction for this problem is from independent set in planar graphs 
of maximum degree $3$. Then we argue that this problem can be embedded
using boxes whose points lie in two parallel planes. One plane is used
to define the triangular-grid-like structure and 
the other is used to encode the subset of vertices that describe the
induced subgraph of the grid.

\subsection{Triangular Grid}
Let $\Gamma$ be the infinite graph with vertex set and edge set (see Figure~\ref{fig:triangulargrid1})
\begin{align*}
	V(\Gamma) ~&=~ \big\{ ( i+ j\cdot 1/2, j\cdot\sqrt{3}/2) \mid i,j\in \NN \big\}, \\
	E(\Gamma) ~&=~ \left \{ ab\mid a,b\in V(\Gamma),~ 
		\text{the Euclidean distance between $a$ and $b$ is exactly $1$} \right\}.
\end{align*}
First we show that the following intermediate problem, which is closely related to 
independent set, is NP-hard.

\begin{figure}
	\centering
	\includegraphics[page=1,scale=.8]{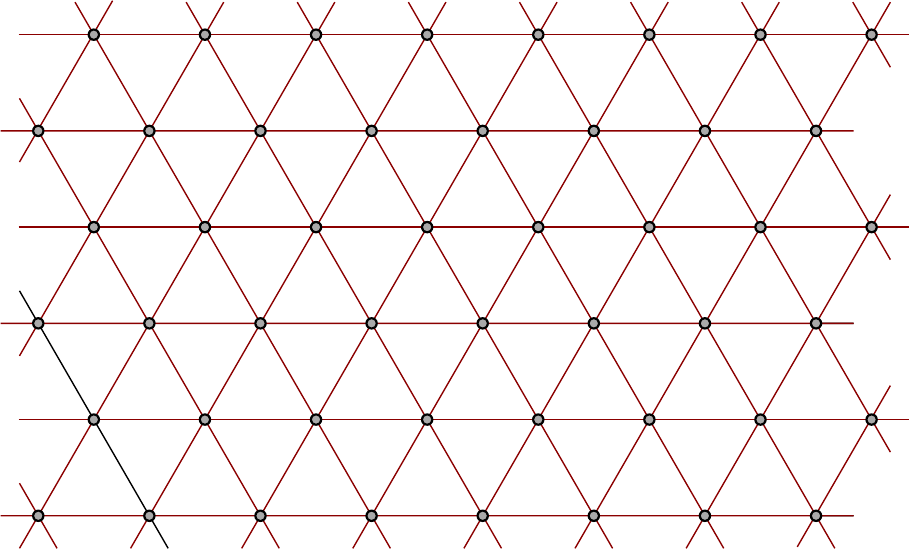}
	\caption{Triangular grid $\Gamma$.}
	\label{fig:triangulargrid1}
\end{figure}

\begin{quote}
	{\sc Independent Set on Induced Triangular Grid}\\
	\emph{Input:} A pair $(A,\ell)$, where $A$ is a subset of $V(\Gamma)$
		and $\ell$ is a positive integer.\\
	\emph{Question:} Is there a subset $B\subseteq A$ of size $\ell$ such that
		no two vertices in $B$ are connected by an edge of $E(\Gamma)$?
\end{quote}

\begin{lemma}
\label{lem:Independent Set Induced Triangular Grid}
	{\sc Independent Set on Induced Triangular Grid} is NP-complete.
\end{lemma}
\begin{proof}
	It is obvious that the problem is in NP.

	Garey and Johnson~\cite{gj77} show that the problem {\sc Vertex Cover} is NP-complete 
	for planar graphs of degree at most $3$. 
	Since a subset $U \subseteq V(G)$ is a vertex cover of graph $G$ if and only if $V(G)\setminus U$
	is an independent set of $G$, it follows that the problem {\sc Independent Set} is NP-complete 
	for planar graphs of degree at most $3$. For the rest of the proof,
	let $G$ be a planar graph of degree at most $3$.
	
	Let us define a \DEF{$\Gamma$-representation} of $G$ to be a pair $(H,\varphi)$, 
	where $H\subset \Gamma$ and $\varphi$ is a mapping, with the following
	properties:
	\begin{itemize}
		\item Each vertex $u$ of $G$ is mapped to a distinct vertex $\varphi(u)$ of $H$.
		\item Each edge $uv$ of $G$ is mapped to a simple path $\varphi(uv)$ contained in $H$
			and connecting $\varphi(u)$ to $\varphi(v)$.
		\item For each two distinct edges $uv$ and $u'v'$ of $G$, 
			the paths $\varphi(uv)$ and $\varphi(u'v')$ are disjoint except at the common endpoints
			$\{ \varphi(u),\varphi(v)\}\cap \{ \varphi(u'),\varphi(v')\}$.
		\item The graph $H$ is precisely
			the union of $\varphi(u)$ and $\varphi(uv)$ over all vertices $u$
			and edges $uv$ of~$G$.
	\end{itemize}
	Note that if $(H,\varphi)$ is a $\Gamma$-representation of $G$ then
	$H$ is a subdivision of $G$. The map $\phi$ identifies which parts of $H$ correspond
	to which parts of $G$.
	
 	A planar graph $G$ with $n$ vertices and maximum degree $3$ (and also $4$)
	can be drawn in a square grid of polynomial size, 
	and such a drawing can be obtained in polynomial time, see, e.g., the results by Storer~\cite{s84} or by Tamassia and Tollis~\cite{tt89}.
	Applying the shear mapping $(x,y)\mapsto (x+y/2,y\sqrt{2}/3)$ to the plane, 
	the square grid becomes a subgraph of $\Gamma$.
	Therefore, we can obtain a $\Gamma$-representation $(H_1,\varphi_1)$ 
	of $G$ of polynomial size. Note that we only use edges of $\Gamma$ that
	are horizontal or have positive slope; edges of $\Gamma$ with negative slope are not used.
	
	Next, we obtain another $\Gamma$-representation $(H_2,\varphi_2)$ 
	such that $H_2$ is an \emph{induced} subgraph of $\Gamma$.
	Induced means that two vertices of $H_2$ are connected with an edge in $H_2$ 
    if and only if the edge exists in $\Gamma$.
	For this, we first scale up the $\Gamma$-representation $(H_1,\varphi_1)$ by a factor $2$ so that
	each edge of $H_1$ becomes a 2-edge path. The new vertices used in the subdivision
	have degree $2$ and its $2$ incident edges have the same orientation.
	After the subdivision, vertices of degree $3$ look like in Figure~\ref{fig:triangulargrid2}.
	Scaling up the figure by a factor of $3$, and rerouting within a small neighbourhood
	of each vertex $v$ that was already in $H_1$, we obtain a 
	$\Gamma$-representation $(H_2,\varphi_2)$ 
	such that $H_2$ is an \emph{induced} subgraph of $\Gamma$.
	See Figure~\ref{fig:triangulargrid2} for an example of such a local transformation.

	\begin{figure}
		\centering
		\includegraphics[page=2,width=.8\textwidth]{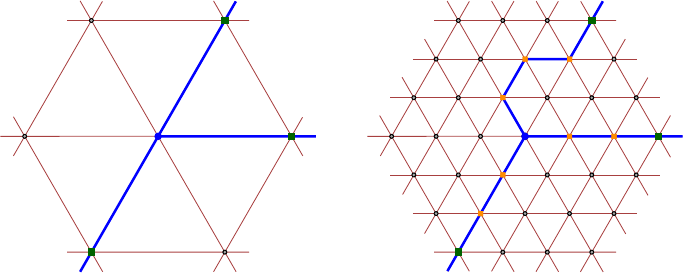}
		\caption{Transformation to get an induced subgraph of the triangular grid.
			Vertices from the subdivision of edges are green squares.}
		\label{fig:triangulargrid2}
	\end{figure}
	
	Now we have a $\Gamma$-representation $(H_2,\varphi_2)$ such that $H_2$ is
	an induced subgraph of $\Gamma$. We want to obtain another $\Gamma$-representation where for each edge $uv\in E(G)$ the path
	$\varphi_2(uv)$ uses an even number of interior edges. For this, we can slightly reroute
	each path $\varphi_2(uv)$ that has an odd number of interior points, see Figure~\ref{fig:triangulargrid3}. 
	To make sure that the graph is still induced, we can first scale up the situation by 
	a factor $2$, and then reroute all the edges $\varphi_2(uv)$ that use an odd number of
	interior vertices. (This is actually all the edges $uv\in E(G)$	because of the scaling.)
	Let $(H_3,\varphi_3)$ be the resulting $\Gamma$-representation of $G$.
	Note that $H_3$ is an induced subgraph of $\Gamma$ and it is a subdivision
	of $G$ where each edge is subdivided an even number of times.

	\begin{figure}
		\centering
		\includegraphics[page=3,scale=1.5]{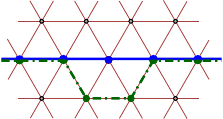}
		\caption{Choosing the parity of paths.}
		\label{fig:triangulargrid3}
	\end{figure}
	
	Let $\alpha(G)$ denote the size of the largest independent set in $G$.
	For each edge $uv$ of $G$, let $2k_{uv}$ be the number of internal
	vertices in the path $\varphi_3(uv)$.
	Then $\alpha(H_3)=\alpha(G)+\sum_{uv\in E(G)} k_{uv}$.
	Indeed, we can obtain $H_3$ from $G$ by repeatedly replacing an edge by a $3$-edge path, i.e., making 2 subdivisions on the same edge. Moreover, any such replacement increases the size of the largest independent set by exactly 1.
	
	It follows that the problem {\sc Independent Set} is NP-complete
	in \emph{induced} subgraphs of the triangular grid $\Gamma$.
	This is precisely the problem {\sc Independent Set on Induced Triangular Grid},
	where we take $A$ to be the set of vertices defining the induced subgraph.
\end{proof}

\subsection{The Point Set}
Let $m\ge 3$ be an arbitrary integer and consider the point set $P_m$ defined by (see Figure~\ref{fig:grid1})
\[
	P_m ~=~ \{ (x,y,z)\in \NN^3\mid x+y+z=m\}.
\]

\begin{figure}
	\centering
	\includegraphics[page=1,scale=.8]{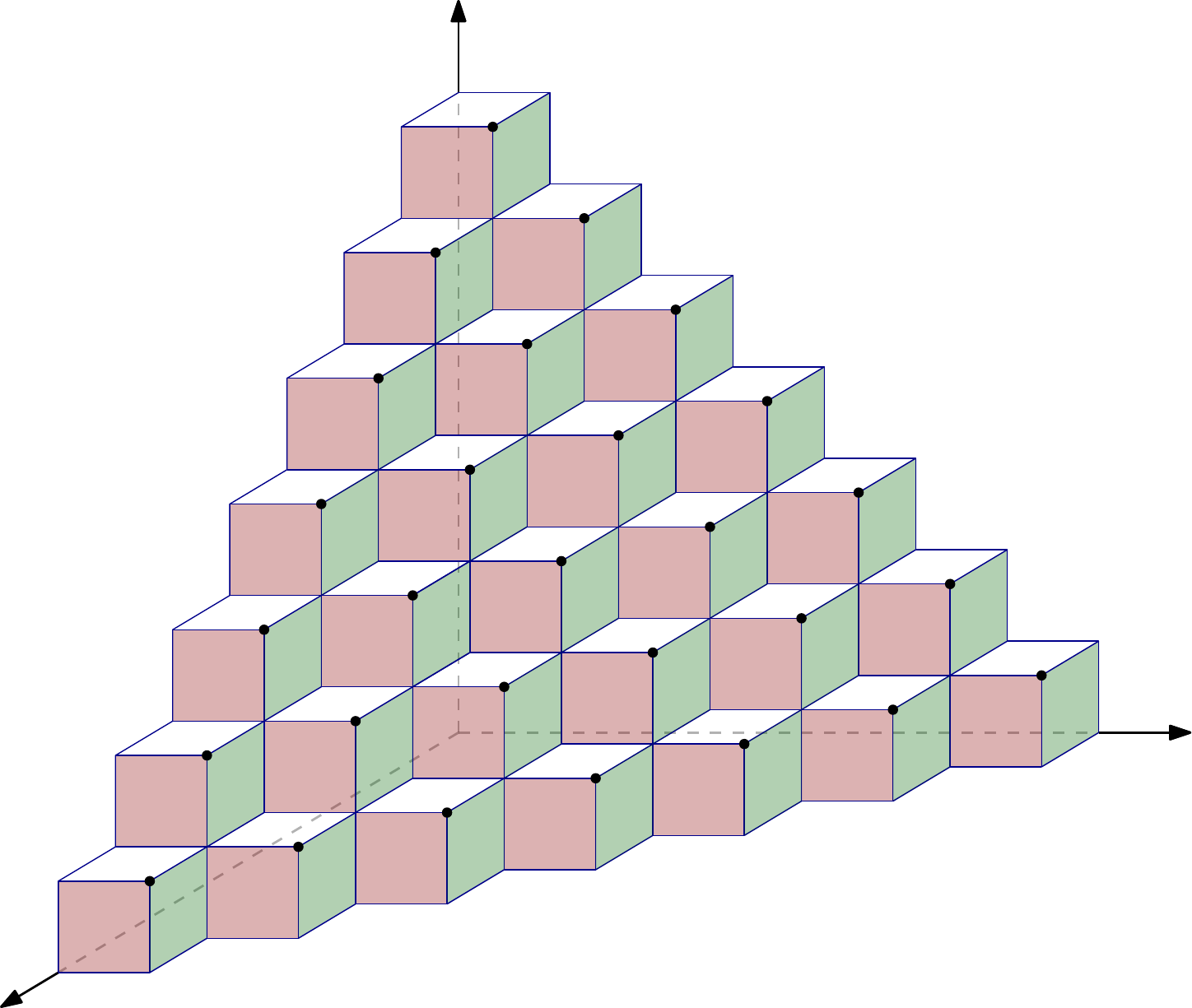}
	\caption{The point set $P_m$ and the boxes $\bbox(p)$, with $p\in P_m$, for $m=9$.}
	\label{fig:grid1}
\end{figure}

\noindent
Standard induction shows that
the set $P_m$ has $1+2+\dots+(m-2)=(m-1)(m-2)/2$ points
and that 
\[
	\mu(P_m) ~=~ \vol \left( \bigcup_{p\in P_m} \bbox(p) \right) ~=~
	m(m-1)(m-2)/6.
\]
This last number appears as sequence A000292, tetrahedral (or triangular pyramidal) numbers,
in~\cite{seq}.

Consider the real number $\eps= 1/4m^2$, 
and define the vector $\Delta_\eps=(\eps,\eps,\eps)$.
Note that $\eps$ is much smaller than $1$.
For each point $p\in P_{m-1}$, consider the point $p+\Delta_\eps$, see Figure~\ref{fig:grid2}.
Let us define the set $Q_m$ to be 
\[
	Q_m ~=~ \{ p+\Delta_\eps \mid p\in P_{m-1} \}.
\]
It is clear that $Q_m$ has $|P_{m-1}|=(m-2)(m-3)/2$ points, for $m\ge 3$.
The points of $Q_m$ lie on the plane $x+y+z=m-1+3\eps$.

\begin{figure}
	\centering
	\includegraphics[page=2,scale=.7]{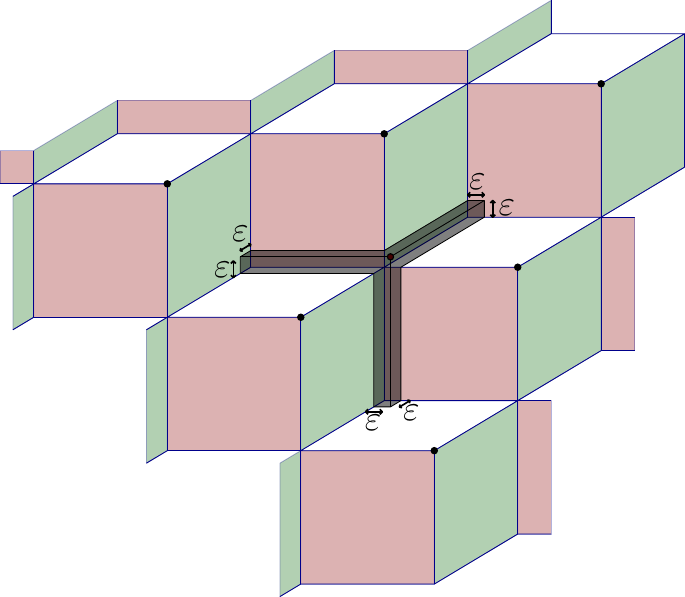}
	\caption{The point $q=p+\Delta_\eps$ and the set $\diff(q)$.}
	\label{fig:grid2}
\end{figure}

For each point $q$ of $Q_m$ define 
\[
	\diff(q) ~=~ \unio\big(P_m \cup \{q\}\big) \setminus \unio\big(P_m\big) ~=~ \left(\bigcup_{p\in P_m\cup \{q \}} \bbox(p)\right)\setminus
	 \left( \bigcup_{p\in P_m} \bbox(p)\right).
\]
Note that $\diff(q)$ is the union of $3$ boxes of size $\eps\times \eps\times 1$
and a cube of size $\eps\times\eps\times\eps$, see Figure~\ref{fig:grid2}.
To get the intuition for the following lemma, 
see Figure~\ref{fig:grid3}.

\begin{figure}
	\centering
	\includegraphics[page=3,scale=.8]{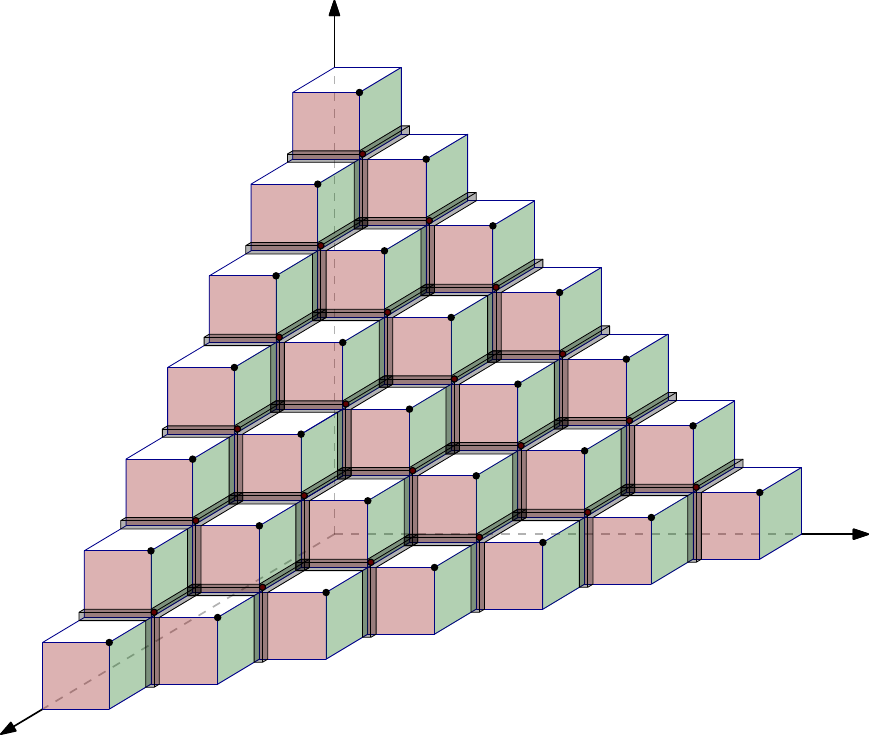}
	\caption{The sets $\diff(q)$ for all $q\in Q_m$.}
	\label{fig:grid3}
\end{figure}

\begin{lemma}
\label{lem:reduction1}
	Consider any $Q'\subseteq Q_m$.
	\begin{itemize}
	\item If the sets $\diff(q)$, for all $q\in Q'$, are pairwise disjoint, then
		$\mu(P_m\cup Q')=\mu(P_m) + |Q'|\cdot (3 \eps^2 + \eps^3)$. 
	\item If $Q'$ contains two points $q_0$ and $q_1$ such that $\diff(q_0)$ and $\diff(q_1)$ intersect, then
		$\mu(P_m\cup Q')<\mu(P_m) + |Q'|\cdot (3 \eps^2 + \eps^3)$.
	\end{itemize}
\end{lemma}
\begin{proof}
	Note that for each $q\in Q_m$ we have
	\[
		\mu(P_m\cup \{q \})-\mu(P_m) ~=~ \vol(\diff(q)) ~=~
		3 \eps^2 + \eps^3.
	\]

	If the sets $\{ \diff(q)\mid q\in Q'\}$ are pairwise disjoint
	then
	\begin{align*}
		\mu(P_m\cup Q') ~&=~ \mu(P_m) +\vol\left( \bigcup_{q\in Q'} \diff(q)\right) \\
		&=~ \mu(P_m) + \sum_{q\in Q'} \vol(\diff(q)) \\
		&=~ \mu(P_m) + |Q'| \cdot \left(3 \eps^2 + \eps^3\right).
	\end{align*}
	
	Consider now the case when $Q'$ contains two points $q_0$ and $q_1$ 
	such that $\diff(q_0)$ and $\diff(q_1)$ intersect.
	The geometry of the point set $Q'$ implies
	that $\diff(q_0)$ and $\diff(q_1)$ intersect 
	in a cube of size $\eps\times\eps\times\eps$, see Figure~\ref{fig:grid3}.
	Therefore, we have
	\begin{align*}
		\mu(P_m\cup Q') ~&=~ 
		\mu(P_m) + 
		\vol\left( \bigcup_{q\in Q'} \diff(q)\right) \\
		&\le~ \mu(P_m) + \sum_{q\in Q'} \vol(\diff(q)) - \vol(\diff(q_0)\cap \diff(q_1)) \\
		&=~ \mu(P_m) + |Q'| \cdot \left(3 \eps^2 + \eps^3\right) - \eps^2 \\
		&< \mu(P_m) + |Q'| \cdot \left(3 \eps^2 + \eps^3\right). \tag*{\qedhere}
	\end{align*}
\end{proof}

We can define naturally a graph $T_m$ on the set $Q_m$ 
by using the intersection of the sets $\diff(\cdot)$.
The vertex set of $T_m$ is $Q_m$, and two points $q,q'\in Q_m$ define
an edge $qq'$ of $T_m$ if and only if $\diff(q)$ and $\diff(q')$ intersect, see Figure~\ref{fig:grid4}. Simple geometry shows that $T_m$ is
isomorphic to a part of the triangular grid $\Gamma$.
Thus, choosing $m$ large enough, we can get an arbitrarily large portion of the triangular grid $\Gamma$.
Note that a subset of vertices $Q'\subseteq Q_m$ is independent in $T_m$ if and only if 
the sets $\left\{ \diff(q)\mid q\in Q' \right)$ are pairwise disjoint.

\begin{figure}
	\centering
	\includegraphics[page=4,scale=.8]{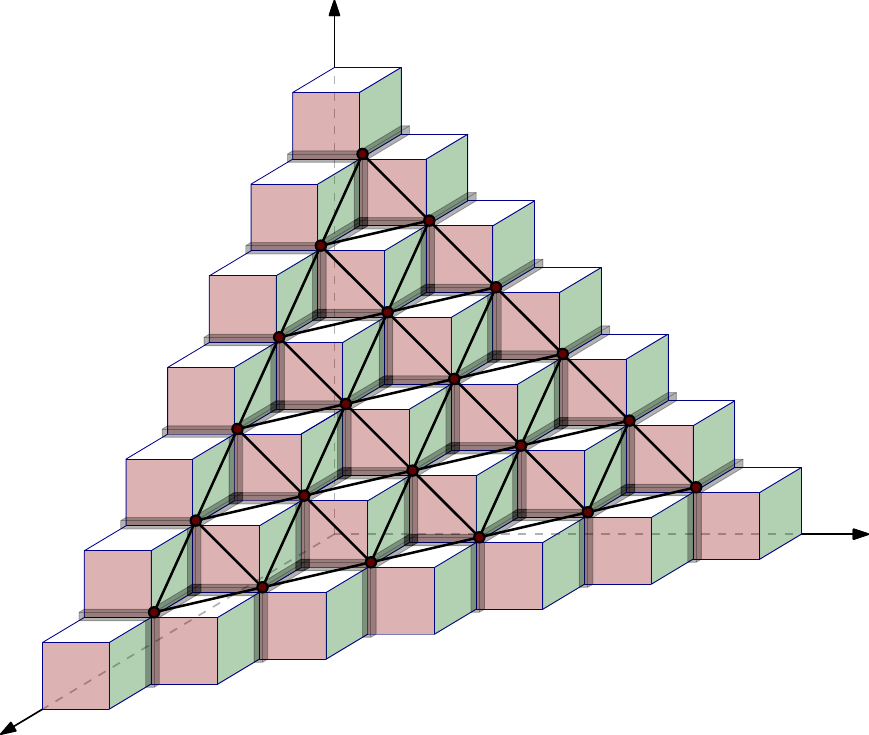}
	\caption{The graph $T_m$ for $m=9$.}
	\label{fig:grid4}
\end{figure}

We next show that picking points in $P_m$ has higher priority than picking points in $Q_m$.

\begin{lemma}
\label{lem:allP}
	Let $P'$ be a subset of $P_m$ such that $P_m\setminus P'$ is not empty.
	Then $\mu(P'\cup Q_m)< \mu(P_m)$.
\end{lemma}
\begin{proof}
	Assume that $P_m\setminus P'$ contains exactly one point, denoted by $p$.
	Having a smaller set $P'$ can only decrease the value of $\mu(P'\cup Q_m)$.
	Then 
	\[ 
		\mu(P') ~=~ \mu(P_m)-1 .
	\]

	Consider the sets of $3$ points 
	\begin{align*}
		Q^1_m(p)~=~\{ &(p_x-1,p_y,p_z)+\Delta_\eps , 
			(p_x,p_y-1,p_z)+\Delta_\eps , 
			(p_x,p_y,p_z-1)+\Delta_\eps \}
		~\subseteq~ Q_m, \\
		Q^2_m(p)~=~\{ &(p_x-1,p_y-1,p_z+1)+\Delta_\eps , 
			(p_x+1,p_y-1,p_z-1)+\Delta_\eps ,\\ 
			&(p_x-1,p_y+1,p_z-1)+\Delta_\eps \}
		~\subseteq~ Q_m.		
	\end{align*}
	Figure~\ref{fig:grid5} is useful for the following computations.
	For each point $q\in Q^1_m(p)$ we have
	\[
		\mu(P'\cup q) ~=~ \mu(P')+ \vol(\diff(q)) + \eps.
	\]
	For each point $q\in Q^2_m(p)$ we have
	\[
		\mu(P'\cup q) ~=~ \mu(P')+ \vol(\diff(q)) + \eps^2.
	\]
	Using that $\eps^2\le \eps$ because $0<\eps<1$, we get
	\[
		\forall q\in Q^1_m(p) \cup Q^2_m(p):~~~
			\mu(P'\cup q) ~\le~ \mu(P')+ \vol(\diff(q)) + \eps.
	\]
	For all points $q$ of $Q_m\setminus (Q^1_m(p)\cup Q^2_m(p))$ we have
	\[
		\mu(P'\cup q) ~=~ \mu(P')+ \vol(\diff(q)).
	\]
	
	\begin{figure}
	\centering
	\includegraphics[page=5,width=\textwidth]{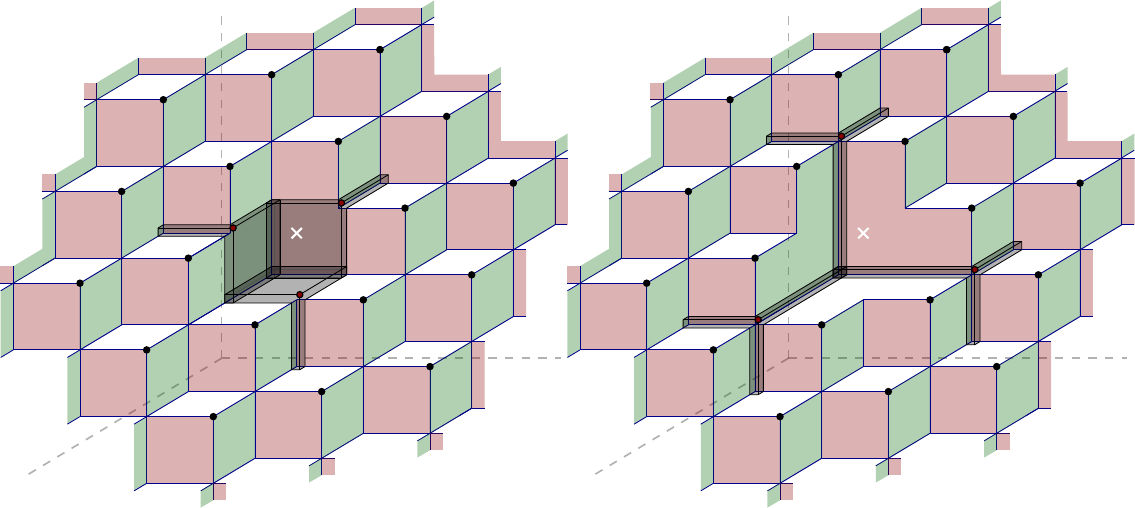}
	\caption{Image for the proof of Lemma~\ref{lem:allP}. The point $p$ of $P_m$
		that is missing in $P'$ is indicated with a white cross. 
		Left: the contribution of points from $Q^1_m(p)$.
		Right: the contribution of points from $Q^2_m(p)$.}
	\label{fig:grid5}
	\end{figure}
	
	\noindent
	We thus have 
	\begin{align*}
		\mu(P'\cup Q_m) ~&\le~ \mu(P')+ \sum_{q\in Q_m} \vol(\diff(q)) 
			+ \sum_{q\in Q^1_m(p)\cup Q^2_m(p)} \eps \\
			&=~ \mu(P_m)-1 + |Q_m| \cdot (3 \eps^2 + \eps^3) 
			+ 6\cdot \eps \\
			&\le~ \mu(P_m)-1 + \frac{(m-2)(m-3)}{2}\cdot 4\cdot \eps 
			+ 6\cdot \eps \\
			&< ~ \mu(P_m),
	\end{align*}
	where the last step uses $\eps= 1/4m^2$.
\end{proof}

\subsection{The Reduction}

We are now ready to prove NP-completeness of {\sc 3d  Volume Selection}.

\begin{theorem}
	The problem {\sc 3d  Volume Selection} is NP-complete.
\end{theorem}
\begin{proof}
	It is obvious that the problem is in NP.
	To show hardness we reduce from the problem {\sc Independent Set Induced Triangular Grid},
	shown to be NP-complete in Lemma~\ref{lem:Independent Set Induced Triangular Grid}.
	
	Consider an instance $(A,\ell)$ to {\sc Independent Set on Induced Triangular Grid},
	where $A$ is a subset of the vertices of the triangular grid $\Gamma$
	and $\ell$ is an integer. 
	Take $m$ large enough so that $T_m$ is isomorphic to an induced 
	subgraph of $\Gamma$ that contains $A$. 
	Recall that $\eps=1/4m^2$.
	For each vertex $v$ of $T_m$ let $\psi_{\Gamma}(v)$ be the corresponding vertex of $\Gamma$.
	For each subset $B$ of $A$, let $Q_m(B)$ be the subset of $T_m$ that corresponds to $B$, that is, 
	$Q_m(B)=\{ q\in Q_m\mid \psi_{\Gamma}(q)\in B\}$.
	
	Consider the set of points $P=P_m\cup Q_m(A)$, 
	the parameter $k=(m-1)(m-2)/2+\ell$,
	and the value $V=\frac{m(m-1)(m-2)}{6} + \ell\cdot (3\eps^2+\eps^3)$.
	We claim that $(A,\ell)$ is a yes-instance for {\sc Independent Set on Induced Triangular Grid}
	if and only if $(P,k,V)$ is a yes-instance for {\sc 3d Volume Selection}.
	
	If $(A,\ell)$ is a yes-instance for {\sc Independent Set on Induced Triangular Grid},
	there is a subset $B\subseteq A$ of $\ell$ independent vertices in $\Gamma$.
	This implies that $Q_m(B)$ is an independent set in $T_m$, that is,
	the sets $\{ \diff(q)\mid q\in Q_m(B)\}$ are pairwise disjoint.
	Lemma~\ref{lem:reduction1} then implies that 
	\begin{align*}
		\mu(P_m\cup Q_m(B))~&=~ \mu(P_m)+ |B|\cdot (3\eps^2+\eps^3) \\
		&=~ \frac{m(m-1)(m-2)}{6} + \ell \cdot (3\eps^2+\eps^3) \\
		&=~ V.
	\end{align*}
	Therefore $P_m\cup Q_m(B)$ is a subset of $P$ with $|P_m|+|B|=(m-1)(m-2)/2+\ell =k$
	points such that $\mu(P_m\cup Q_m(B))=V$. It follows that 
	$(P,k,V)$ is a yes-instance for {\sc 3d Volume Selection}.
	
	Assume now that $(P,k,V)$ is a yes-instance for {\sc 3d Volume Selection}.
	This means that $P$ contains a subset $Q$ of $k$ points such that
	\[
		\mu(Q)~\ge~ V ~=~ \frac{m(m-1)(m-2)}{6} + \ell\cdot (3\eps^2+\eps^3)
			~=~ \mu(P_m) + \ell\cdot (3\eps^2+\eps^3) ~>~ \mu(P_m).
	\]
	Because of Lemma~\ref{lem:allP}, it must be that $P_m$ is contained in $Q$,
	as otherwise we would have $\mu(Q)<\mu(P_m)$.
	Since we have $P_m\subset Q$ and $P=P_m \cup Q_m(A)$, we obtain that
	$Q$ is $P_m\cup Q_m(B)$ for some $B\subseteq A$.
	Moreover, $|B|=k-|P_m|=\ell$.
	By Lemma~\ref{lem:reduction1}, if $Q_m(B)$ is not an independent set in $T_m$,
	we have
	\[
		\mu (Q) ~=~ \mu (P_m\cup Q_m(B)) ~<~ \mu (P_m)+ \ell (3\eps^2 +\eps) ~=~ V,
	\]
	which contradicts the assumption that $\mu (Q)\ge V$.
	Therefore it must be that $Q_m(B)$ is an independent set in $T_m$.
	It follows that $B\subset A$ has size $\ell$ and forms an independent set in $\Gamma$,
	and thus $(A,\ell)$ is a yes-instance for {\sc Independent Set on Induced Triangular Grid}.
\end{proof}

\section{Exact Algorithm in 3 Dimensions}
\label{sec:exact}

In this section we design an algorithm to solve \volsel\ in 3 dimensions in time $n^{O(\sqrt{k})}$. The main insight is that, for an optimal solution $Q^*$,
the boundary of $\unio(Q^*)$ is a planar graph with $O(k)$ vertices, and therefore
has a balanced separator with $O(\sqrt{k})$ vertices. We would like to guess the separator,
break the problem into two subproblems, and solve each of them recursively.
This basic idea leads to a few technical challenges to take care of.
One obstacle is that subproblems should be really independent because we do not want 
to double count some covered parts. 
Essentially, a separator in the graph-theory sense does not imply
independent subproblems in our context.
Another technicality is that some of the subproblems that we encounter recursively 
cannot be solved optimally; we can only get a lower bound to the optimal value. 
However, for the subproblems that define the optimal solution at the higher
level of the recursion, we do compute an optimal solution.

Let $P$ be a set of $n$ points in the positive quadrant of $\RR^3$.
Through our discussion, we will assume that $P$ is fixed and thus
drop the dependency on $P$ and $n$ from the notation.
We can assume that no point of $P$ is dominated by
another point of $P$. 
Using an infinitesimal perturbation of the points,
we can assume that all points have all coordinates different.
Indeed, we can replace each point $p$ by the point $p+i(\eps,\eps,\eps)$,
where $i$ is a different integer for each point of $P$ and $\eps>0$ is
an infinitesimal value or a value that is small enough.

Let $M$ be the largest $x$- or $y$-coordinate in $P$, 
thus $M=\max\{p_x,p_y\mid p\in P\}$.
We define $\sigma$ to be the square in $\RR^2$
defined by $[-1,M+1]\times [-1,M+1]$.
It has side length $M+2$.

For each subset $Q$ of $P$, consider the projection of $\unio(Q)$ 
onto the $xy$-plane. This defines a plane graph, which we denote by $G(Q)$,
and which we define precisely in the following, see Figure~\ref{fig:exact1}.
We consider $G(Q)$ as a geometric, embedded graph where each vertex
is a point and each edge is (drawn as) a straight-line segment, in fact,
a horizontal or vertical straight-line segment on the $xy$-plane.
There are different types of vertices in $G(Q)$.
The projection of each point $q\in Q$ defines a vertex, which we denote by $v_q$.
When for two distinct points $q,q'\in Q$
the boundary of the projection of the boxes $\bbox(q)$
and the boundary of the projection of $\bbox(q')$ intersect outside the 
$x$- and $y$-axis, then they do so exactly once because of our 
assumption on general position, and this defines a vertex
that we denote by $v_{q,q'}$. (Not all pairs $(q,q')$ define such a vertex.)
Additionally, each point $q\in Q$ defines a vertex $v_{x,q}$ at position $(q_x,0)$
and a vertex $v_{y,q}$ at position $(0,q_y)$. 
Finally, we have a vertex $v_{x,y}$ placed at the origin.
The vertices of $G(Q)$ are connected in a natural way: the boundary
of the visible part of $\bbox(q)$ connects the points that appear 
on that boundary. In particular, the vertices on the $x$-axis
are connected and so do those on the $y$-axis.
Since we assume general position, each vertex uniquely determines the
boxes that define it. 
Each vertex $q\in Q$ defines a bounded face $f(q,Q)$ in $G(Q)$.
This is the projection of the face on the boundary of $\unio(Q)$
contained in the plane $\{ (x,y,z)\in \RR^3\mid z=q_z\}$, see Figure~\ref{fig:exact1}, right.
In fact, each bounded face of $G(Q)$ is $f(q,Q)$ for some $q\in Q$.

\begin{figure}
	\centering
	\includegraphics[page=1,scale=1.4]{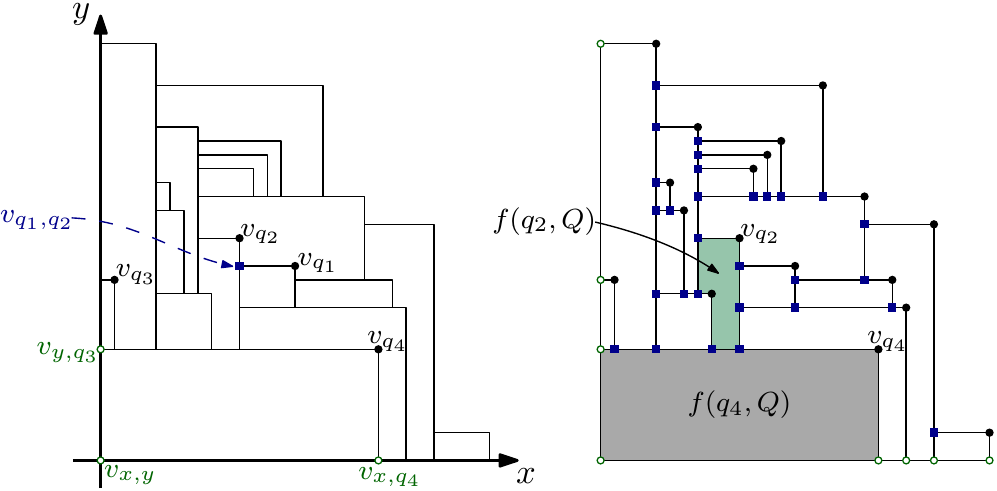}
	\caption{A sample of the different vertices in $G(Q)$ and the faces of $G(Q)$.}
	\label{fig:exact1}
\end{figure}

We triangulate each bounded face $f(q,Q)$ of $G(Q)$ \emph{canonically}, as follows, see Figure~\ref{fig:exact2}.
The boundary of a bounded face $f(q,Q)$ is made of 
a top horizontal segment $t(q,Q)$ (which may contain several edges of the graph), 
a right vertical segment $r(q,Q)$ (which may contain several edges of the graph),
and a monotone path $\gamma(q,Q)$ from the top, left corner to the bottom, right
corner. Such a monotone path $\gamma(q,Q)$ alternates horizontal and vertical segments
and has non-decreasing $x$-coordinates and non-increasing $y$-coordinates.
Let $v_t(q,Q)$ be the first interior vertex of $\gamma(q,Q)$ and
let $v_r(q,Q)$ be the last interior vertex of $\gamma(q,Q)$.
Note that $v_q$ is the vertex where $t(q,Q)$ and $r(q,Q)$ meet. 
We add diagonals from $v_q$ to all interior vertices of $\gamma(q,Q)$,
diagonals from $v_t(q,Q)$ to all the interior vertices of $t(q,Q)$
and diagonals from $v_r(q,Q)$ to all the interior vertices of $r(q,Q)$. 
This is the canonical triangulation of the face $f(q,Q)$,
and we apply it to each bounded face of $G(Q)$.

\begin{figure}
	\centering
	\includegraphics[page=2,scale=1]{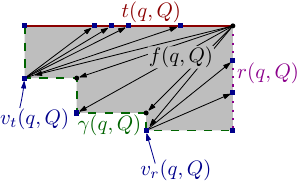}
	\caption{Triangulating a bounded face of $G(Q)$.}
	\label{fig:exact2}
\end{figure}

The outer face of $G(Q)$ may also have many vertices. 
We place on top the square $\sigma$, with vertices $\{ -1,M+1\}^2$.
From the vertices at $(-1,-1)$ and $(M+1,M+1)$ 
we add all possible edges, while keeping planarity.
From the vertex $(-1,M+1)$ we add the edges to $(-1,-1)$, to $(M+1,M+1)$,
and to the highest vertex on the $y$-axis.
Similarly, from the vertex $(M+1,-1)$ we add the edges to $(-1,-1)$, to $(M+1,M+1)$,
and to the rightmost vertex on the $x$-axis.
With such an operation, the outer face is defined by the boundary of the square $\sigma$.
 
Let $T(Q)$ be the resulting geometric, embedded graph, see Figure~\ref{fig:exact3}.
The graph $T(Q)$ is a triangulation of the square $\sigma$ with internal vertices.
It is easy to see that $G(Q)$ and $T(Q)$ have $O(|Q|)$ vertices and edges.
For example, one can argue that $G(Q)$ has $|Q|+1$ faces and no parallel edges,
and the graph $T(Q)$ is a triangulation of $G(Q)$ with $4$ additional vertices.
To treat some extreme cases, we also define $T(\emptyset)=\sigma$, as a graph,
with the diagonal of positive slope.

\begin{figure}
	\centering
	\includegraphics[page=3,scale=1.2]{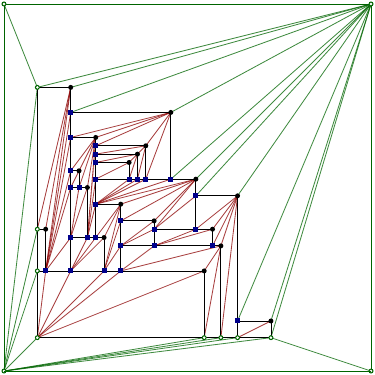}
	\caption{The graph $T(Q)$.}
	\label{fig:exact3}
\end{figure}

A polygonal domain is a subset of the plane defined by a polygon where we
remove the interior of some polygons, which form holes. 
The combinatorial complexity of a domain $D$, denoted by $|D|$,
is the number of vertices and edges used to define it.
We say that a polygonal curve or a family of polygonal curves in $\RR^2$ 
is \DEF{$Q$-compliant} if 
the edges of of the curves are also edges of $T(Q)$.
A polygonal domain $D$ is \DEF{$Q$-compliant}
if its boundary is $Q$-compliant. 
Note that a $Q$-compliant polygonal domain has combinatorial complexity $O(|Q|)$
because the graph $T(Q)$ has $O(|Q|)$ edges.

Consider a set $Q\subseteq P$ and a $Q$-compliant polygonal curve $\gamma$.
Let $P_\gamma$ be the points of $P$ that participate in the definition
of the vertices on $\gamma$. 
Thus, if $v_q$ is in $\gamma$,
we add $q$ to $P_\gamma$; if $v_{q,q'}$ is in $\gamma$, 
we add $q$ and $q'$ to $P_\gamma$; 
if $v_{x,q}$ is in $\gamma$, we add $q$ to $P_\gamma$, and so on.
Since each vertex on $\gamma$ contributes $O(1)$ vertices
to $P_\gamma$, we have $|P_\gamma|=O(|\gamma|)$.
For a family $\Gamma$ of polygonal curves we define 
$P_\Gamma=\cup_{\gamma\in \Gamma} P_\gamma$.
For a polygonal domain $D$ with boundary $\partial D$
we then have $|P_{\partial D}| = O(|D|)$.

\begin{lemma}
\label{lem:compliant}
	If $\gamma$ is a $Q$-compliant polygonal curve
	then, for each $Q'$ with $P_\gamma\subseteq Q'\subset Q$,
	the curve $\gamma$ is also $Q'$-compliant.
\end{lemma}	
\begin{proof}
	For each edge $e$ of $T(Q)$, the edge $e$ is also
	contained in $T(\tilde Q)$ for all $\tilde Q$ that contain~$P_e$.
	It follows that $T(Q')$ has all the edges $e$ contained in $\gamma$,
	and thus $T(Q')$ contains $\gamma$.
\end{proof}

We are going to use dynamic programming based on planar separators of $T(Q^*)$ for an
optimal solution $Q^*$. 
A \DEF{valid tuple} to define a subproblem is a tuple $(S,D,\ell)$, where 
$S\subset P$,
$D$ is an $S$-compliant polygonal domain, and
$\ell$ is a positive integer.
The tuple $(S,D,\ell)$ models a subproblem where the points of $S$ are
already selected to be part of the feasible solution, $D$ is a $S$-compliant domain so that
we only care about the volume inside the cylinder $D\times \RR$,
and we can still select $\ell$ points from $P\cap (D\times \RR)$.
We have two different values associated to each valid tuple, depending on 
which subsets $Q$ of vertices from $P\cap D$ can be selected:
\begin{align*}
	\optfree(S,D,\ell) ~=~ \max\{ &\vol(\unio(S\cup Q)\cap (D\times \RR)) \mid 
		Q\subset P\cap (D\times \RR),~|Q|\le \ell \}.\\
	\optcomp(S,D,\ell) ~=~ \max\{ &\vol(\unio(S\cup Q)\cap (D\times \RR)) \mid 
		Q\subset P\cap (D\times \RR),~|Q|\le \ell, \\ 
		& \text{$D$ is $(S\cup Q)$-compliant}\}.	
\end{align*}
Obviously, we have for all valid tuples $(S,D,\ell)$
\[
	\optcomp(S,D,\ell)~\le~ \optfree(S,D,\ell).
\]
On the other hand, we are interested in the valid tuple $(\emptyset,\sigma,k)$,
for which we have 
$\optfree(\emptyset,\sigma,k)=\optcomp(\emptyset,\sigma,k)$.

We would like to get a recursive formula for 
$\optfree(S,D,\ell)$ or $\optcomp(S,D,\ell)$ using planar separators.
More precisely, we would like to use a separator in $T(S\cup Q^*)$ for
an optimal solution, and then branch on all possible such separators.
However, none of the two definitions seem good enough for this.
If we would use $\optfree(S,D,\ell)$, then we divide into 
domains that may have too much freedom and the interaction between
subproblems gets complex.
If we would use $\optcomp(S,D,\ell)$, then merging the problems
becomes an issue. Thus, we take a mixed route where
we argue that, for the valid tuples that are relevant 
for finding the optimal solution, we actually have $\optfree=\optcomp$.

We start showing how to compute
$\optcomp(S,D,\ell)$ in the obvious way. 
This will be used to solve the base cases of the recursion.
\begin{lemma}
\label{lem:bottomrecursion}
	We can compute $\optcomp(S,D,\ell)$ 
	in $O(n^{\ell+2} )$ time.
\end{lemma}
\begin{proof}
	We enumerate all the subsets $Q$ of $P\cap D$ with $\ell$ points,
	and for each such $Q$ we proceed as follows.
	We first build $T(S\cup Q)$
	and check whether $D$ is contained in the edge set of $T(S\cup Q)$.
	If it is not, then $D$ is not $(S\cup Q)$-compliant and we move
	to the next subset $Q$.
	Otherwise, we compute $\unio(S\cup Q)$, its restriction to $D\times \RR$,
	and its volume. Standard approaches can be used to do this
	in $O((|S|+|Q|+|D|)^2)= O(n^2)$ time, 
	for example working with the projection onto the $xy$-plane.
	(The actual degree of the polynomial is not relevant.)
	This procedure enumerates $O(|P|^\ell)=O(n^\ell)$ subsets of $P$ and 
	for each one spends $O(n^2)$ time. The result follows.
\end{proof}

A \DEF{valid partition} $\pi$ of $(S,D,\ell)$ is a
collection of valid tuples 
$\pi=\left\{ (S_1,D_1,\ell_1),\dots, (S_t,D_t,\ell_t) \right\}$
such that
\begin{itemize}
	\item $S_1=\dots=S_t=S\cup S_0$ for some set $S_0\subset P\cap D$;
	\item $|S_0| = O\left( \sqrt{ |S|+\ell } \right)$;
	\item the domains $D_1$,$\dots$, $D_t$ have pairwise disjoint interiors
		and $D=\bigcup_i D_i$; 
	\item $\ell= |S_0|+ \sum_i \ell_i$; and
	\item $\ell_i\le 2\ell/3$ for each $i=1,\dots,t$.
\end{itemize}
Let $\Pi(S,D,\ell)$ be the family of valid partitions for the tuple $(S,D,\ell)$.
We remark that different valid partitions may have different cardinality.

\begin{lemma}
\label{lem:recursionfree}
	For each valid tuple $(S,D,\ell)$ we have
	\[
		\optfree(S,D,\ell)~\ge~ \max_{\pi \in \Pi (S,D,\ell)}~
			\sum_{(S',D',\ell')\in \pi} \optfree (S',D',\ell').
	\]
\end{lemma}
\begin{proof}
	For any valid partition $\pi\in \Pi(S,D,\ell)$,
	let $S_\pi$ be the smallest set such that 
	$S'=S\cup S_\pi$ for all tuples $(S',D',\ell')\in \pi$. 
	This means that $S_\pi= S'\setminus S$ for an arbitrary $(S',D',\ell')\in \pi$.
	For each such tuple $(S',D',\ell')\in \pi$, let $Q^*(S',D',\ell')$ be an
	optimal solution to $\optfree(S',D',\ell')$, and define 
	\[
		Q_\pi= S_\pi \cup \bigcup_{(S',D',\ell')\in \pi} Q^*(S',D',\ell').
	\]
	Then from the properties of valid partitions we have
	\[
		|Q_\pi| ~=~ |S_\pi| + \sum_{(S',D',\ell')\in \pi} |Q^*(S',D',\ell')| ~=~ 
		|S_\pi| + \sum_{(S',D',\ell')\in \pi} \ell' ~=~ \ell.
	\]
	Obviously, $Q_\pi$ is contained in $P\cap D$ because $D$ contains $P\cap D_i$.
	
	We have seen that for each valid partition $\pi\in \Pi(S,D,\ell)$
	the set $Q_\pi$ is a feasible solution considered in
	the problem $\optfree(S,D,\ell)$.
	Therefore 
	\[ \optfree(S,D,\ell) ~\ge~ \max_{\pi \in \Pi(S,D,\ell)} \vol(\unio(S\cup Q_\pi)\cap (D\times \RR)). \]
	Using that the interiors of $\{ D'\mid (S',D',\ell')\in\pi\}$ are pairwise disjoint, and then using that $S'\cup Q^*(S',D',\ell')$ is contained in $S\cup Q_{\pi}$ for all $(S',D',\ell')\in\pi$, we obtain
	\begin{align*}
		\optfree(S,D,\ell) ~&\ge~ \max_{\pi \in \Pi(S,D,\ell)} \sum_{(S',D',\ell')\in\pi} 
				\vol(\unio(S\cup Q_\pi)\cap (D'\times \RR))\\ 
			&\ge~ \max_{\pi \in \Pi(S,D,\ell)} \sum_{(S',D',\ell')\in\pi} 
				\vol(\unio(S'\cup Q^*(S',D',\ell'))\cap (D'\times \RR)).
	\end{align*}	
	Since $Q^*(S',D',\ell')$ is optimal for $\optfree(S',D',\ell')$, we obtain the desired
	\[ \optfree(S,D,\ell) ~\ge~  \max_{\pi \in \Pi(S,D,\ell)} \sum_{(S',D',\ell')\in\pi} \optfree(S',D',\ell'). \qedhere
	\]
\end{proof}

\begin{lemma}
\label{lem:recursioncompliant}
	For each valid tuple $(S,D,\ell)$ we have
	\[
		\optcomp(S,D,\ell)~\le~ \max_{\pi \in \Pi (S,D,\ell)}~
			\sum_{(S',D',\ell')\in \pi} \optcomp (S',D',\ell').
	\]
\end{lemma}
\begin{proof}
	Let $Q^*$ be the optimal solution defining $\optcomp(S,D,\ell)$.
	Thus, $Q^*\subseteq P\cap D$ has at most $\ell$ points, 
	$D$ is $(S\cup Q^*)$-compliant,	and
	\[
		\optcomp(S,D,\ell) ~=~ \vol(\unio(S\cup Q^*)\cap (D\times \RR)).
	\]
	Consider the triangulation $T(S\cup Q^*)$. This is a $3$-connected planar graph.
	Recall that the boundary of $D$ is contained in $T(S\cup Q^*)$ because $D$ is
	$(S\cup Q^*)$-compliant. 
	Note that $T(S\cup Q^*)$ has $O(|S \cup Q^*|)= O(|S|+\ell)$ vertices.
	
	Assign weight $1/|Q^*|$ to the vertices $v_q$, $q\in Q^*$,
	and weight $0$ to the rest of vertices in $T(S\cup Q^*)$. 
	The sum of the weights is obviously $1$.
	Because of the cycle-separator theorem of Miller~\cite{Miller86},
	there is a cycle $\gamma$ in $T(S\cup Q^*)$ with $O(\sqrt{|S|+\ell})$ vertices,
	such that the interior of $\gamma$ has at most $2|Q^*|/3$ vertices of $Q^*$ and
	the exterior of $\gamma$ has at most $2|Q^*|/3$ vertices of $Q^*$. 
	
	Since $\gamma$ has $O(\sqrt{|S|+\ell})$ vertices, the set $P_\gamma$ also
	has $O(\sqrt{|S|+\ell})$ vertices. Note that $P_\gamma\subseteq S\cup Q^*$.
	Take $S_0=P_\gamma\setminus S$, so that $S\cup P_\gamma$ is the disjoint
	union of $S$ and $S_0$.
	Because of Lemma~\ref{lem:compliant}, the 
	domain $D$ and the cycle $\gamma$ are $(S\cup S_0)$-compliant.

	The cycle $\gamma$ breaks the domain $D$ into at least $2$ domains.
	Let $\D= \{ D_1,\dots,D_t\}$ be those domains.
	Since the boundary of each domain $D_i\in \D$ is contained in $\partial D\cup \gamma$,
	each domain $D_i\in \D$ is $(S\cup S_0)$-compliant.
	For each domain $D_i\in \D$, let 
	$Q^*_i=\{ q\in Q^* \setminus (S\cup S_0)\mid v_q\in D_i\}$
	and let $\ell_i=|Q^*_i|$. 
	Since the interior of $D_i$ is either in the interior or the exterior of $\gamma$, 
	we have $\ell_i\le 2\ell/3$	for each $D_i\in \D$.
	Moreover, $|\ell| = |S_0| + \sum_i \ell_i$ because the points $q$ of $Q^*$
	that could be counted twice have the corresponding vertex $v_q$ on $\gamma$, 
	but then they also belong to $P_\gamma\subset S\cup S_0$ 
	and thus cannot belong to $Q_i^*$.
	
	The properties we have derived imply that
	$\pi_\gamma =\{ (S\cup S_0,D_i,\ell_i)\mid i=1,\dots,t \}$ is a valid partition of $(S,D,\ell)$,
	and thus $\pi_\gamma\in \Pi(S,D,\ell)$.
	Moreover $Q^*_i$ is a feasible solution for the problem 
	$\optcomp(S\cup S_0,D_i,\ell_i)$.
	Indeed, since $D_i$ is $(S\cup S_0)$-compliant and $(S\cup Q^*)$-compliant,
	Lemma~\ref{lem:compliant} implies that $D_i$ is also 
	$(S\cup S_0\cup Q^*_i)$-compliant.
	
	Note that, for each $(S\cup S_0,D_i,\ell_i)$ in the partition $\pi_\gamma$
	we have
	\begin{equation}
	\label{eq:2}
		\vol(\unio(S\cup Q^*) \cap (D_i\times \RR)) ~=~
		\vol(\unio(S\cup S_0\cup Q^*_i) \cap (D_i\times \RR)).
	\end{equation}
	Indeed, for a point $q\in Q^*\setminus (S\cup S_0\cup Q^*_i)$,
	$\bbox(q)$ may contribute to the union 
	$\unio(S\cup P_\gamma\cup Q^*)$, but when projected
	onto the $xy$-plane it lies outside the domain $D_i$
	because the face $f(q,S\cup Q^*)$ lies outside $D_i$.

	Therefore we obtain
	\begin{align*}
	  \optcomp(S,D,\ell) ~&=~ \vol(\unio(S\cup Q^*)\cap (D\times \RR)) ~\le~ \sum_i \vol(\unio(S\cup Q^*)\cap (D_i\times \RR)),
	\end{align*}
	where we used $D=\bigcup_i D_i$. With equation \eqref{eq:2}, and then using that $Q_i^*$ is feasible for $\optcomp (S\cup S_0 ,D_i,\ell_i)$, we get
	\begin{align*}
	  \optcomp(S,D,\ell) ~&\le~ \sum_i \vol(\unio(S\cup S_0 \cup Q^*_i)\cap (D_i\times \RR)) \\
	  &\le~ \sum_i \optcomp (S\cup S_0 ,D_i,\ell_i) ~=~ \sum_{(S',D',\ell')\in \pi_\gamma } \optcomp (S',D',\ell').
	\end{align*}
	The statement now follows since $\pi_\gamma\in \Pi(S,D,\ell)$.
\end{proof}

Our dynamic programming algorithm closely follows the inequality
of Lemma~\ref{lem:recursioncompliant}.
Specifically, we define for each valid tuple $(S,D,\ell)$ the value
\[
	\dpcomp (S,D,\ell) ~=~ \begin{cases}
				\optcomp (S,D,\ell) &\text{if $\ell\le O(\sqrt{k})$;}\\
				{\displaystyle \max_{\pi \in \Pi (S,D,\ell)}~
					\sum_{(S',D',\ell')\in \pi} \dpcomp (S',D',\ell')}, 
						&\text{otherwise.}\\
				\end{cases}
\]

\begin{lemma}
\label{lem:relation}
	For each valid tuple $(S,D,\ell)$ we have
	\[ 
		\optcomp(S,D,\ell) ~\le~ \dpcomp(S,D,\ell) 
		~\le~ \optfree(S,D,\ell).
	\]
\end{lemma}
\begin{proof}
	We show this by induction on $\ell$.
	When $\ell\le O(\sqrt{k})$, then from the definitions 
	we have
	\[ 
		\dpcomp(S,D,\ell) ~=~ \optcomp(S,D,\ell) ~\le~ \optfree(S,D,\ell).
	\]
	This covers the base cases.
	For larger values of $\ell$, we use Lemma~\ref{lem:recursioncompliant},
	the induction hypothesis, and the definition of $\dpcomp(\cdot)$
	to derive
	\begin{align*}
		\optcomp(S,D,\ell)~&\le~ \max_{\pi \in \Pi (S,D,\ell)}~
			\sum_{(S',D',\ell')\in \pi} \optcomp (S',D',\ell') \\
		&\le~
		\max_{\pi \in \Pi (S,D,\ell)}~
			\sum_{(S',D',\ell')\in \pi} \dpcomp (S',D',\ell') \\
		&=~ \dpcomp(S,D,\ell).
	\end{align*}
	Also for larger values of $\ell$, we use 
	the definition of $\dpcomp(\cdot)$, the induction hypothesis, and 
	Lemma~\ref{lem:recursionfree},
	to derive
	\begin{align*}
		\dpcomp(S,D,\ell)~&=~ \max_{\pi \in \Pi (S,D,\ell)}~
			\sum_{(S',D',\ell')\in \pi} \dpcomp (S',D',\ell') \\
		&\le~
		\max_{\pi \in \Pi (S,D,\ell)}~
			\sum_{(S',D',\ell')\in \pi} \optfree (S',D',\ell') \\
		&\le~ \optfree(S,D,\ell). \tag*{\qedhere}
	\end{align*}
\end{proof}

Since we know that 
$\optfree(\emptyset,\sigma,k)=\optcomp(\emptyset,\sigma,k)$,
Lemma~\ref{lem:relation} implies that 
$\dpcomp(\emptyset,\sigma,k)=\optfree(\emptyset,\sigma,k)$.
Hence, it suffices to compute $\dpcomp(\emptyset,\sigma,k)$ using its recursive
definition.
In the remainder, we bound the running time of this algorithm.
	
\begin{theorem}
\label{thm:exact}
    In 3 dimensions, \volsel\ can be solved in time $n^{O(\sqrt{k})}$.
\end{theorem}
\begin{proof}
	We compute $\dpcomp(\emptyset,\sigma,k)$
	using its recursive definition.
	We need a bound on the number of different subproblems,
	defined by valid tuples $(S,D,\ell)$ that appear in the recursion.
	We will see that there are $n^{O(\sqrt{k})}$ different subproblems.
	
	Starting with $(S_1,D_1,\ell_1)=(\emptyset, \sigma,k)$,
	consider a sequence of valid tuples $(S_1,D_1,\ell_1)$, $(S_2,D_2,\ell_2)$, 
	$\dots$ such that, for $i\ge 2$, 
	the tuple $(S_{i},D_{i},\ell_{i})$ appears in some
	valid partition of $(S_{i-1},D_{i-1},\ell_{i-1})$.
	Because of the properties of valid partitions,
	we have $\ell_i\le 2\ell_{i-1}/3$
	and $|S_{i-1}|\le |S_i| \le |S_{i-1}| + O(\sqrt{|S_i|+\ell_{i-1}})$.
	
	Let $i_0$ be the first index $i$ with $|S_i| > \ell_i$. 
	Consider first the indices $i < i_0$, where $|S_i|\le \ell_i$.
	Then $|S_i| \le |S_{i-1}| + O(\sqrt{\ell_{i-1}})$
	and it follows by induction that 
	\begin{align*}
		|S_i| ~&\le~ |S_1| + O(\sqrt{\ell_{1}}) + O(\sqrt{\ell_{2}}) + \dots + O(\sqrt{\ell_{i-1}})\\
		&\le~ 0 + O\bigg( \sum_{j< i} \sqrt{\ell_j} \bigg)
		~\le~ O\Bigg( \sum_{j< i} \sqrt{\left( \frac{2}{3}\right)^j \ell_1} \Bigg)
		~\le~ O\Big( \sqrt{\ell_1} \Big)
		~\le~ O\Big( \sqrt{k} \Big),
	\end{align*}
	where we have used that $\ell_1=k$.
	By definition of $i_0$,
	for $i> i_0$ we have $|S_i|\le |S_{i_0}|+\ell_{i_0} \le 2 |S_{i_0}|= O(\sqrt{k})$.
	Therefore, for \emph{all} indices $i$ we have 
	$|S_i|=O(\sqrt{k})$.
	
	For each valid tuple that appears in the recursive computation
	of $\dpcomp(\emptyset,\sigma,k)$, there is some
	sequence of valid tuples, as considered before, that contains it.
	It follows that, for \emph{all} valid tuples 
	$(S,D,\ell)$ considered through the algorithm we have $|S|=O(\sqrt{k})$.
	
	Let us give an upper bound on the valid tuples $(S,D,\ell)$ 
	that appear in the computation.
	There are $n^{O(\sqrt{k})}$ choices for the set $S$.
	Once we have fixed $S$, the domain $D$ has to be $S$-compliant, and
	this means that we have to select edges in the triangulated graph $T(S)$.
	Since $T(S)$ has $O(|S|)=O(\sqrt{k})$ vertices and edges, there are at
	most $2^{|E(T(S))|}= 2^{O(\sqrt{k})}$ possible choices for $D$.
	Finally, we have $k$ options for the value $\ell$.
	Therefore, there are at most 
	\[
		n^{O(\sqrt{k})} \cdot 2^{O(\sqrt{k})} \cdot k ~=~ n^{O(\sqrt{k})} 
	\]
	valid tuples $(S,D,\ell)$ that appear in the recursion.
	
	We next bound how much time we spend for each tuple.
	Consider a valid tuple $(S,D,\ell)$ that appears through the recursion.
	If $\ell=O(\sqrt{k})$, we compute $\dpcomp(S,D,\ell)$ 
	using Lemma~\ref{lem:bottomrecursion} in $n^{O(\ell)}=n^{O(\sqrt{k})}$ time.
	Otherwise, to compute $\dpcomp(S,D,\ell)$ we have to iterate over all the valid
	partitions $\Pi(S,D,\ell)$. There are $n^{O(\sqrt{k})}$ 
	such valid partitions. Indeed, we have to select the subset $S_0\subset D\cap P$
	with $O(\sqrt{k})$ vertices and then the partitioning of $D$ 
	into regions $D_1,\dots,D_t$ that are $(S\cup S_0)$-compliant. 
	This can be bounded by $n^{O(\sqrt{k})}$.
	(Alternatively, we can iterate over the $n^{O(\sqrt{k})}$ 
	possible options to define the separating cycle~$\gamma$ used in the proof 
	of Lemma~\ref{lem:recursioncompliant}.)
	
	We conclude that in the computation of $\dpcomp(\emptyset,\sigma,k)$ 
	we have to consider $n^{O(\sqrt{k})}$ valid tuples and
	for each one of them computing $\dpcomp(\cdot)$ takes $n^{O(\sqrt{k})}$ time.
	The result follows.
\end{proof}

We only described an algorithm that computes $\HSS(P,k)$, i.e., the maximal volume realized by any size-$k$ subset of $P$. 
It is easy to augment the algorithm with appropriate bookkeeping to also compute an actual optimal subset.

\section{Efficient Polynomial-Time Approximation Scheme}
\label{sec:ptas}

In this section we design an approximation algorithm for \volsel. 

\begin{theorem}
  Given a point set $P$ of size $n$ in $\RR_{>0}^d$, $0 \le k \le n$, and $0<\eps \le 1/2$, we can compute a $(1 \pm \eps)$-approximation of $\HSS(P,k)$ in time $O(n \cdot \eps^{-d} (\log n + k + 2^{O(\eps^{-2} \log 1/\eps)^d}))$. We can also compute a set $S \subseteq P$ of size at most $k$ such that $\mu(S)$ is a $(1 - \eps)$-approximation of $\HSS(P,k)$ in the same time.
\end{theorem}

We also discuss an improvement to time $O\big( 2^{O(\eps^{-2} \log 1/\eps)^d} \cdot n \log n \big)$ in Section~\ref{sec:improvementtim}.

The approach is based on the shifting technique of Hochbaum and Maass~\cite{HochbaumM85}.
However, there are some non-standard aspects in our application. 
It is impossible to break the problem into independent subproblems because all
the anchored boxes intersect around the origin. We instead break
the input into subproblems that are \emph{almost} independent.
To achieve this, we use an exponential grid, instead of the usual regular
grid with equal-size cells. Alternatively, this could be interpreted
as using a regular grid in a $\log$-$\log$ plot of the input points.

Throughout this section we need two numbers $\lambda,\tau \approx d/\eps$. Specifically, we define $\tau$ as the smallest integer larger than $d/\eps$, and $\lambda$ as the smallest power of $(1-\eps)^{-1/d}$ larger than $d/\eps$.
We consider a partitioning of the positive quadrant $\RR_{>0}^d$ into \DEF{regions} of the form
\[ R(\bx) := \prod_{i=1}^d [\lambda^{x_i},\lambda^{x_i+1}) \quad \text{for} \quad \bx = (x_1,\ldots,x_d) \in \ZZ^d. \]
On top of this partitioning we consider a grid, where each grid cell contains $(\tau-1)^d$ regions and the grid boundaries are thick, i.e., two grid cells do not touch but have a region in between. More precisely, for any offset $\bl = (\ell_1,\ldots,\ell_d) \in \ZZ^d$, we define the grid \DEF{cells} 
\[ C_\bl(\by) := \prod_{i=1}^d [\lambda^{\tau \cdot y_i + \ell_i + 1},\lambda^{\tau(y_i+1) + \ell_i}) \quad \text{for} \quad \by = (y_1,\ldots,y_d) \in \ZZ^d. \]
Note that each grid cell indeed consists of $(\tau-1)^d$ regions, and the space not contained in any grid cell (i.e., the grid boundaries) consists of all regions $R(\bx)$ with $x_i \equiv \ell_i \pmod \tau$ for \emph{some} $1 \le i \le d$.

Our approximation algorithm now works as follows (cf.\ the pseudocode given below). 

(1) Iterate over all grid offsets $\bl \in [\tau]^d$. This is the key step of the shifting technique of Hochbaum and Maass~\cite{HochbaumM85}.

(2) For any choice of the offset $\bl$, remove all points not contained in any grid cell, i.e., remove points contained in the thick grid boundaries. This yields a set $P' \subseteq P$ of remaining points. 

(3) The grid cells now induce a partitioning of $P'$ into sets $P_1',\ldots,P_m'$, where each $P_i'$ is the intersection of $P'$ with a grid cell $C_i$ (with $C_i = C_\bl(\by^{(i)})$ for some $\by^{(i)} \in \ZZ^d$).
Note that these grid cell subproblems $P_1',\ldots,P_m'$ are not independent, since any two boxes have a common intersection near the origin, no matter how different their coordinates are. However, we will see that we may treat $P_1',\ldots,P_m'$ as independent subproblems since we only want an approximation. 

(4) We discretize by rounding down all coordinates of all points in $P_1',\ldots,P_m'$ to powers of\footnote{Here we use that $\lambda$ is a power of $(1-\eps)^{-1/d}$, to ensure that rounded points are contained in the same cells as their originals.} $(1-\eps)^{1/d}$. We can remove duplicate points that are rounded to the same coordinates. This yields sets $\tilde P_1,\ldots,\tilde P_m$. Note that within each grid cell in any dimension the largest and smallest coordinate differ by a factor of at most $\lambda^{\tau-1}$. Hence, there are at most $\log_{(1-\eps)^{-1/d}}(\lambda^{\tau-1}) = O(\eps^{-2} \log 1/\eps)$ different rounded coordinates in each dimension, and thus the total number of points in each $\tilde P_i$ is $O(\eps^{-2} \log 1/\eps)^d$. 

(5) Since there are only few points in each $\tilde P_i$, we can precompute all \volsel\ solutions on each set $\tilde P_i$, i.e., for any $1 \le i \le m$ and any $0 \le k' \le |\tilde P_i|$ we precompute $\HSS(\tilde P_i,k')$. We do so by exhaustively enumerating all $2^{|\tilde P_i|}$ subsets $S$ of $\tilde P_i$, and for each one computing $\mu(S)$ by inclusion-exclusion in time $O(2^{|S|})$ (see, e.g.,~\cite{while2006faster,wu2001metrics}). This runs in total time $O(m \cdot 2^{O(\eps^{-2} \log 1/\eps)^d}) = O(n \cdot 2^{O(\eps^{-2} \log 1/\eps)^d})$.

(6) It remains to split the $k$ points that we want to choose over the subproblems $\tilde P_1,\ldots,\tilde P_m$. As we treat these subproblems independently, we compute 
\[ V(\bl) := \max_{k_1+\ldots+k_m \le k} \; \sum_{i=1}^m \HSS(\tilde P_i,k_i). \]
Note that if the subproblems would be independent, then this expression would yield the exact result. We argue below that the subproblems are sufficiently close to being independent that this expression yields a $(1-\eps)$-approximation of $\HSS(\bigcup_{i=1}^m \tilde P_i,k)$.
Observe that the expression $V(\bl)$ can be computed efficiently by dynamic programming, where we compute for each $i$ and $k'$ the following value:
\[ T[i,k'] = \max_{k_1+\ldots+k_i \le k'}\; \sum_{i'=1}^i \HSS(\tilde P_{i'},k_{i'}). \]
The following rule computes this table (see the pseudocode below for further details):
\[ T[i,k'] = \max_{0 \le \kappa \le \min\{k',|\tilde P_i|\}} \big( \HSS(\tilde P_i,\kappa) + T[i-1,k'-\kappa] \big). \]

(7) Finally, we optimize over the offset $\bl$ by returning the maximal $V(\bl)$.

\medskip

This finishes the description of the approximation algorithm. In pseudocode, this yields the following procedure.
\begin{enumerate}[label=(\arabic*)]
  \item Iterate over all offsets $\bl = (\ell_1,\ldots,\ell_d) \in [\tau]^d$:
  \begin{enumerate}[label=(\arabic*)] \setcounter{enumii}{1}
    \item $P' := P$. Delete any $p$ from $P'$ that is not contained in any grid cell $C_\bl(\by)$.
    \item Partition $P'$ into $P_1',\ldots,P_m'$, where $P_i' = P' \cap C_i$ for some grid cell $C_i$.
    \item Round down all coordinates to powers of $(1-\eps)^{1/d}$ and remove duplicates, obtaining $\tilde P_1,\ldots,\tilde P_m$.
    \item Compute $H[i,k'] := \HSS(\tilde P_i,k')$ for all $1 \le i \le m$, $0 \le k' \le |\tilde P_i|$.
    \item Compute $V(\bl) := \max_{k_1+\ldots+k_m \le k} \sum_{i=1}^m \HSS(\tilde P_i,k_i)$ by dynamic programming:
    \begin{enumerate}[label=--] 
      \item Initialize $T[i,k'] = 0$ for all $0 \le i \le m$, $0 \le k' \le k$.
      \item For $i=1,\ldots,m$, for $\kappa=0,\ldots,|\tilde P_i|$, and for $k'=\kappa,\kappa+1,\ldots,k$:
      \begin{enumerate}[label=--] 
        \item Set $T[i,k'] := \max\{T[i,k'], H[i,\kappa] + T[i-1,k'-\kappa]\}$
      \end{enumerate}
      \item Set $V(\bl) := T[m,k]$.
    \end{enumerate}
  \end{enumerate} \setcounter{enumi}{6}
  \item Return $\max_\bl V(\bl)$.
\end{enumerate}

\subsection{Running Time}
Step (1) yields a factor $\tau^d = O(\frac 1\eps)^d$ in the running time. Since we can compute for each point in constant time the grid cell it is contained in, step (2) runs in time $O(n)$. For the partitioning in step (3), we use a dictionary data structure storing all $\by \in \ZZ^d$ with nonempty $P' \cap C_\bl(\by)$. Then we can assign any point $p \in P'$ to the other points in its cell by one lookup in the dictionary, in time $O(\log n)$. Thus, step (3) can be performed in time $O(n \log n)$. Step (4) immediately works in the same running time. For step (5) we already argued above that it can be performed in time $O\big(n 2^{O(\eps^{-2} \log 1/\eps)^d}\big)$. Finally, from the pseudocode for step (6) we read off a running time of $O(\sum_{i=1}^m |\tilde P_i| \cdot k) = O(n k)$. The total running time is thus
\[ O\Big(n \cdot \eps^{-d} \big(\log n + k + 2^{O(\eps^{-2} \log 1/\eps)^d}\big)\Big). \]

\subsection{Correctness}
The following lemmas show that the above algorithm indeed computes a $(1\pm O(\eps))$-approximation of $\HSS(P)$. Reducing $\eps$ by an appropriate constant factor then yields a $(1\pm\eps)$-approximation.

\begin{lemma}[Removing grid boundaries] \label{lem:removegridboundary}
  Let $P$ be a point set and let $0 \le k \le |P|$. Remove all points contained in grid boundaries with offset $\bl$ to obtain the point set $P_\bl := P \cap \bigcup_{\by \in \ZZ^d} C_\bl(\by)$. 
  Then for \emph{all} $\bl \in \ZZ^d$ we have
  \[ \HSS(P_\bl,k) \le \HSS(P,k), \]
  and for \emph{some} $\bl \in \ZZ^d$ we have
  \[ \HSS(P_\bl,k) \ge (1-\eps) \HSS(P,k). \]
\end{lemma}
\begin{proof}
  Since we only remove points, the first inequality is immediate. 
  For the second inequality we use a probabilistic argument. 
  Consider an optimal solution, i.e., a set $S \subseteq P$ of size at most $k$ with $\mu(S) = \HSS(P,k)$. Let $S_\bl := S \cap P_\bl$. 
  For a uniformly random offset $\bl \in [\tau]^d$, consider the probability that a fixed point $p \in S$ survives, i.e., we have $p \in S_\bl$. Consider the region $R(\bx) = \prod_{i=1}^d [\lambda^{x_i},\lambda^{x_i+1})$ containing point $p$, where $\bx = (x_1,\ldots,x_d) \in \ZZ^d$. Recall that the grid boundaries consist of all regions $R(\bx)$ with $x_i \equiv \ell_i \pmod \tau$ for some $1 \le i \le d$. For a uniformly random~$\bl$, for fixed $i$ the equation $x_i \equiv \ell_i \pmod \tau$ holds with probability $1/\tau$. By a union bound, the probability that at least one of these equations holds for $1 \le i \le d$ is at most $d/\tau \le \eps$ (by definition of $\tau$ as the smallest integer larger than $d/\eps$). Hence, $p$ survives with probability at least $1-\eps$. 
  
  Now for each point $q \in \unio(S)$ identify a point $s(q) \in S$ dominating $q$. Since $s(q)$ survives in $S_\bl$ with probability at least $1-\eps$, the point $q$ is dominated by $S_\bl$ with probability at least $1-\eps$. 
  By integrating over all $q \in \unio(S)$ we thus obtain an expected volume of
  \[ \Ex_\bl[\mu(S_\bl)] = \int_{\unio(S)} \Pr[q \text{ is dominated by } S_\bl] d q \ge \int_{\unio(S)} (1-\eps) d q = (1-\eps) \mu(S). \]
  It follows that for some $\bl$ we have $\mu(S_\bl) \ge \Ex[\mu(S_\bl)] \ge (1-\eps) \mu(S)$. 
  For this $\bl$ we have
  \[ \HSS(P_\bl,k) \ge \mu(S_\bl) \ge (1-\eps) \mu(S) = (1-\eps) \HSS(P,k), \]
  where the first inequality uses $|S_\bl| \le k$ and the definition of $\HSS$ as maximizing over all subsets, and the last inequality holds since we picked $S$ as an optimal solution, realizing $\HSS(P,k)$.
\end{proof}

\begin{lemma}[Rounding down coordinates] \label{lem:rounding}
  Let $P$ be a point set, and let $\tilde P$ be the same point set after rounding down all coordinates to powers of $(1-\eps)^{-1/d}$. Then for any $k$
  \[ (1-\eps) \HSS(P,k) \le \HSS(\tilde P,k) \le \HSS(P,k). \]
\end{lemma}
\begin{proof}
  Let $\hat P$ be the set $P$ with all coordinates scaled down by a factor $\alpha := (1-\eps)^{1/d}$. By a simple scaling invariance, we have $\HSS(\hat P,k) = \alpha^d \cdot \HSS(P,k) = (1-\eps) \HSS(P,k)$. Note that for any point $\tilde p \in \tilde P$ the corresponding point $p \in P$ dominates~$\tilde p$, and the corresponding point $\hat p \in \hat P$ is dominated by $\tilde p$.
  Now pick any subset $\tilde S$ of $\tilde P$ of size~$k$, and let $S,\hat S$ be the corresponding subsets of $P,\hat P$. Then we have $\unio(\hat S) \subseteq \unio(\tilde S) \subseteq \unio(S)$,
  which implies $\mu(\hat S) \le \mu(\tilde S) \le \mu(S)$, and thus 
  \[ (1-\eps) \HSS(P,k) = \HSS(\hat P,k) \le \HSS(\tilde P,k) \le \HSS(P,k). \qedhere \]
\end{proof}

In the proof of the next lemma it becomes important that we have used the thick grid boundaries, with a separating region, when defining the grid cells.

\begin{lemma}[Treating subproblems as independent I] \label{lem:independentsubproblemsI}
  For any offset $\bl$, let $S_1,\ldots,S_m$ be point sets contained in different grid cells with respect to offset $\bl$. Then we have
  \[ (1-\eps) \sum_{i=1}^m \mu(S_i) \le \mu\Big( \bigcup_{i=1}^m S_i \Big) \le \sum_{i=1}^m \mu(S_i). \]
\end{lemma}
\begin{proof}
  The second inequality is essentially the union bound. Specifically, for any sets $X_1,\ldots,X_m$ the volume of $\bigcup_{i=1}^m X_i$ is at most the sum over all volumes of $X_i$ for $1 \le i \le m$. In particular, this statement holds with $X_i = \unio(S_i)$, which yields the second inequality.
  
  For the first inequality, observe that we obtain the total volume of all points dominated by $S_1 \cup \ldots \cup S_m$ by summing up the volume of all points dominated by $S_i$ but not by any $S_j$, $j < i$, for each $1 \le i \le m$, i.e., we have
  \begin{equation} \label{eq:claimunioone}
   \mu\Big( \bigcup_{i=1}^m S_i \Big) = \sum_{i=1}^m \bigg( \mu(S_i) - \vol\Big( \unio(S_i) \cap \bigcup_{j<i} \unio(S_j) \Big) \bigg). 
  \end{equation}
  
  Now let $C_\bl(\by^{(i)})$ be the grid cell containing $P_i$ for $1 \le i \le m$, where $\by^{(i)} = (y^{(i)}_1,\ldots,y^{(i)}_d) \in \ZZ^d$. We may assume that these cells are ordered in non-decreasing order of $y^{(i)}_1 + \ldots + y^{(i)}_d$. Observe that in this ordering, for any $j<i$ we have $y^{(j)}_t < y^{(i)}_t$ for \emph{some} $1 \le t \le d$. Recall that $C_\bl(\by) = \prod_{t=1}^d [\lambda^{\tau \cdot y_t + \ell_t + 1},\lambda^{\tau(y_t+1) + \ell_t})$. It follows that each point in $\bigcup_{j<i} \unio(S_j)$ has $t$-th coordinate at most $\delta_t := \lambda^{\tau \cdot y_t + \ell_t}$ for \emph{some} $1 \le t \le d$. Setting $D_t := \{ (z_1,\ldots,z_d) \in \RR_{\ge 0}^d \mid z_t \le \delta_t \}$, we thus have $\bigcup_{j<i} \unio(S_j) \subseteq \bigcup_{t=1}^d D_t$, which yields
  \begin{equation} \label{eq:claimunio}
   \vol\Big( \unio(S_i) \cap \bigcup_{j<i} \unio(S_j) \Big) \le \vol\Big( \unio(S_i) \cap \bigcup_{t=1}^d D_t \Big) \le \sum_{t=1}^d \vol\big( \unio(S_i) \cap D_t \big). 
  \end{equation}
  Let $A$ be the $(d-1)$-dimensional volume of the intersection of $\unio(S_i)$ with the plane $x_t = 0$. 
  Since all points in $S_i$ have $t$-th coordinate at least $\lambda^{\tau \cdot y_t + \ell_t + 1} = \lambda \cdot \delta_t$, we have $\mu(S_i) \ge A \cdot \lambda \cdot \delta_t$. Moreover, $\unio(S_i) \cap D_t$ has $d$-dimensional volume $A \cdot \delta_t$. Together, this yields $\vol(\unio(S_i) \cap D_t) \le \mu(S_i) / \lambda$. With (\ref{eq:claimunioone}) and (\ref{eq:claimunio}), we thus obtain
  \[ \mu\Big( \bigcup_{i=1}^m S_i \Big) \ge \sum_{i=1}^m \big( \mu(S_i) - d \cdot \mu(S_i) / \lambda \big) \ge (1-\eps) \sum_{i=1}^m \mu(S_i), \]
  since $\lambda \ge d/\eps$.
\end{proof}

\begin{lemma}[Treating subproblems as independent II] \label{lem:independentsubproblemsII}
  For any offset $\bl$, let $P_1,\ldots,P_m$ be point sets contained in different grid cells, and $k \ge 0$. Set $P := \bigcup_{i=1}^m P_i$. Then we have
  $$ (1-\eps) \cdot \max_{k_1+\ldots+k_m \le k} \sum_{i=1}^m \HSS(P_i,k_i) \le \HSS(P,k) \le \max_{k_1+\ldots+k_m \le k} \sum_{i=1}^m \HSS(P_i,k_i). $$
\end{lemma}
\begin{proof}
  Consider an optimal solution $S$ of $\HSS(P,k)$ and let $S_i := S \cap P_i$ for $1 \le i \le m$. Then by choice of $S$ as an optimal solution, and by Lemma~\ref{lem:independentsubproblemsI}, we have
  \[ \HSS(P,k) = \mu(S) = \mu\Big( \bigcup_{i=1}^m S_i \Big) \le \sum_{i=1}^m \mu(S_i). \]
  Since \HSS\ maximizes over all subsets and $\sum_{i=1}^m |S_i| = |S| \le k$, we further obtain
  \[ \sum_{i=1}^m \mu(S_i) \le \sum_{i=1}^m \HSS(P_i,|S_i|) \le \max_{k_1+\ldots+k_m \le k} \sum_{i=1}^m \HSS(P_i,k_i). \]
  This shows the second inequality. 
  
  For the first inequality, we pick sets $S_1,\ldots,S_m$, where $S_i \subseteq P_i$ for all $i$ and $\sum_{i=1}^m |S_i| \le k$, realizing $\max_{k_1+\ldots+k_m \le k} \sum_{i=1}^m \HSS(P_i,k_i) = \sum_{i=1}^m \mu(S_i)$. We then argue analogously:
  \[ (1-\eps) \max_{k_1+\ldots+k_m \le k} \sum_{i=1}^m \HSS(P_i,k_i) = (1-\eps) \sum_{i=1}^m \mu(S_i) \le \mu\Big( \bigcup_{i=1}^m S_i \Big) \le \HSS(P,k). \qedhere \]
\end{proof}

Note that the above lemmas indeed prove that the algorithm returns a $(1\pm O(\eps))$-approximation to the value $\HSS(P,k)$. In step (2) we delete the points containing the the grid boundaries, which yields an approximation for some choice of the offset $\bl$ by Lemma~\ref{lem:removegridboundary}. As we iterate over all possible choices for $\bl$ and maximize over the resulting volume, we obtain an approximation. In step (4) we round down coordinates, which yields an approximation by Lemma~\ref{lem:rounding}. Finally, in step (6) we solve the problem $\max_{k_1+\ldots+k_m \le k} \sum_{i=1}^m \HSS(\tilde P_i,k_i)$, which yields an approximation to $\HSS(\bigcup_{i=1}^m \tilde P_i,k)$ by Lemma~\ref{lem:independentsubproblemsII}. All other steps do not change the point set or the considered problem. The final approximation factor is $1\pm O(\eps)$.

\subsection{Computing an Output Set}
The above algorithm only gives an approximation for the value $\HSS(P,k)$, but does not yield a subset $S \subseteq P$ of size $k$ realizing this value. However, by tracing the dynamic programming table we can reconstruct the values $k_1+\ldots+k_m \le k$ with $V(\bl) = \sum_{i=1}^m \HSS(\tilde P_i,k_i)$. By storing in step (5) not only the values $H[i,k']$ but also corresponding subsets $\tilde S_{i,k'} \subset \tilde P_i$, we can thus construct a subset $\tilde S = \tilde S_{1,k_1} \cup \ldots \cup \tilde S_{m,k_m}$ with $V(\bl) = \sum_{i=1}^m \mu(\tilde S_{i,k_i})$.   Lemma~\ref{lem:independentsubproblemsI} now implies that 
\[ \mu(\tilde S) \ge (1-\eps) V(\bl). \]
By storing in step (4) for each rounded point an original point, we can construct a set $S$ corresponding to the rounded points $\tilde S$ such that 
\[ \mu(S) \ge \mu(\tilde S) \ge (1-\eps) V(\bl) \ge (1-O(\eps)) \HSS(P,k), \]
and thus $S$ is a subset of $P$ of size at most $k$ yielding a $(1-O(\eps))$-approximation of the optimal volume $\HSS(P,k)$.

Note that we do not compute the exact volume $\mu(S)$ of the output set $S$. Instead, the value~$V(\bl)$ only is a $(1+O(\eps))$-approximation of $\mu(S)$. 
To explain this effect, recall that exactly computing $\mu(T)$ for any given set $T$ takes time $n^{\Theta(d)}$ (under the Exponential Time Hypothesis). As our running time is $O(n^2)$ for any constant $d,\eps$, we cannot expect to compute $\mu(S)$ exactly.

\subsection{Improved Algorithm}
\label{sec:improvementtim}

The following improvement was suggested to us by Timothy Chan. 
For constant $d$ and $\eps$ the algorithm shown above runs in time $O(n (k + \log n))$. The bottleneck for the $O(nk)$-term is step~(6): Given $H_i(k') := \HSS(\tilde P_i,k')$ for all $1 \le i \le m$, $0 \le k' \le |\tilde P_i|$, we want to compute
\[ \max_{k_1+\ldots+k_m \le k} \; \sum_{i=1}^m H_i(k_i). \]
Note that it suffices to compute an $(1+\eps)$-approximation to this value, to end up with an $(1+O(\eps))$-approximation overall.

This problem is an instance of the \emph{multiple-choice 0/1 knapsack problem}, where we are given a budget $W$ and items $j \in S$ with corresponding weights $w_j$ and profits $p_j$, as well as a partitioning $S = S_1 \cup \ldots \cup S_m$, and the task is to compute the maximum $\sum_{j \in T} p_j$ over all sets $T \subseteq S$ satisfying $\sum_{j \in T} w_j \le W$ and $|T \cap S_i| = 1$ for all $1 \le i \le m$.
In order to cast the above problem as an instance of multiple-choice 0/1 knapsack, we simply set $S_i := \{0,1,\ldots,\min\{k,|\tilde P_i|\}\}$ and define $p_j := H_i(j)$ and $w_j = j$ for all $j \in S_i$. We also set $W := k$. Note that now the constraint $\sum_{j \in T} w_j \le W$ corresponds to $k_1+\ldots+k_m \le k$ and the objective $\sum_{j \in T} p_j$ corresponds to $\sum_{i=1}^m H_i(k_i)$.

For the multiple-choice 0/1 knapsack problem there are known PTAS techniques. In particular, in his Master's thesis, Rhee~\cite[Section 4.2]{Rhee15} claims a time bound of $O(m \eps^{-2} \log (m/\eps) \max_j |S_j| + |S| \log |S|)$. In our case, we have $m \le n$ and $|S_j| = \min\{k,|\tilde P_i|\} + 1 = O(\eps^{-2} \log 1/\eps)^d$. Moreover, $|S| \le m \cdot \max_j |S_j|$. This yields a time of 
$O( n \log(n/\eps) \cdot (\eps^{-2} \log 1/\eps)^d )$.

Plugging this solution for step (6) into the algorithm from the previous sections, we obtain time
\[ O\Big(n \cdot \eps^{-d} \big(\log n + \log(n/\eps) \cdot (\eps^{-2} \log 1/\eps)^d + 2^{O(\eps^{-2} \log 1/\eps)^d}\big)\Big). \]
This can be simplified to $O\big( n \big(\log (n/\eps) \cdot \eps^{-3d} \cdot \log^d(1/\eps) + 2^{O(\eps^{-2} \log 1/\eps)^d}\big) \big)$, which is bounded by $O\big( 2^{O(\eps^{-2} \log 1/\eps)^d} \cdot n \log n \big)$.

\section{Conclusions}
\label{sec:conclude}

We considered the volume selection problem, where we are given $n$ points in $\RR_{>0}^d$ and want to select $k$ of them that maximize the volume of the union of the spanned anchored boxes. 
We show: (1) Volume selection is NP-hard in dimension $d=3$ (previously this was only known when $d$ is part of the input). (2) In 3 dimensions, we design an $n^{O(\sqrt{k})}$ algorithm (the previously best was $\Omega\big(\binom{n}{k}\big)$). (3) We design an efficient polynomial time approximation scheme for any constant dimension $d$ (previously only a $(1-1/e)$-approximation was known).

We leave open to improve our NP-hardness result to a matching lower bound under the Exponential Time Hypothesis, e.g., to show that in $d=3$ any algorithm takes time $n^{\Omega(\sqrt{k})}$ and in any constant dimension $d\ge 4$ any algorithm takes time $n^{\Omega(k)}$. Alternatively, there could be a faster algorithm, e.g., in time $n^{O(k^{1-1/d})}$. 
Finally, we leave open to figure out the optimal dependence on $n,k,d,\eps$ of a $(1-\eps)$-approximation algorithm.

Moving away from the applications, one could also study volume selection on general axis-aligned boxes in $\RR^d$, i.e., not necessarily anchored boxes. This problem {\sc General Volume Selection} is an optimization variant of Klee's measure problem and thus might be theoretically motivated. However, {\sc General Volume Selection} is probably much harder than the restriction to anchored boxes, by analogies to the problem of computing an independent set of boxes, which is not known to have a PTAS~\cite{adamaszek2013approximation}. In particular, {\sc General Volume Selection} is NP-hard already in 2 dimensions, which follows from NP-hardness of computing an independent set in a family of congruent squares in the plane~\cite{fowler1981optimal,imai1983finding}.

\paragraph*{Acknowledgements}
This work was initiated during the Fixed-Parameter Computational
Geometry Workshop at the Lorentz Center, 2016. We are grateful to the other participants
of the workshop and the Lorentz Center for their support.
We are especially grateful to G\"unter Rote for several discussions and related work.

{\raggedright
\bibliographystyle{plainurl}
\bibliography{biblio-shorter}
}

\end{document}